\documentclass[a4paper,11pt]{article}
\usepackage[margin=2.5cm]{geometry}

\usepackage[T1]{fontenc}
\usepackage[utf8]{inputenc}
\usepackage{url}
\usepackage{graphicx}
\usepackage{mathtools}
\usepackage{amsmath}
\usepackage{amssymb}
\usepackage{amsthm}
\usepackage{dsfont}
\usepackage{amsfonts}
\usepackage{mathrsfs}
\usepackage{bbm}
\usepackage{enumerate}
\usepackage{graphicx,tikz,pgf,tikz-cd}
\usepackage[pdfborder={0 0 0}]{hyperref}

\newcounter{CN}

\usetikzlibrary{arrows}

\usepackage[all,cmtip]{xy}
\usepackage{amssymb,amsmath,amsthm,mathrsfs}
\DeclareMathAlphabet{\mathpzc}{OT1}{pzc}{m}{it}
\def\RR{{\mathbb R}}
\def\CC{{\mathbb C}}
\def\NN{{\mathbb N}}
\def\ZZ{{\mathbb Z}}

\def\A{{\mathcal A}}
\def\B{{\mathcal B}}

\def\D{{\mathcal D}}

\def\H{{\mathcal H}}
\def\I{{\mathcal I}}
\def\K{{\mathcal K}}

\def\R{{\mathrm R}}

\def\U{{\mathcal U}}

\def\W{{\mathcal W}}

\def\a{\alpha}

\def\d{\delta}

\def\g{\gamma}
\def\G{\Gamma}

\def\i{\iota}

\def\L{{\mathrm L}}

\def\s{\sigma}

\def\t{\tau}

\def\x{\xi}

\def\gt{\mathfrak t}

\def\Ad{{\hbox{\rm Ad\,}}}
\def\Aut{{\mathrm{Aut}}}

\def\End{{\hbox{End}}}

\def\id{{\rm id}}
\def\1{{\mathbbm 1}}
\def\Exp{{\rm Exp}}

\def\uone{{\rm U(1)}}

\def\diff{{\rm Diff}_+}

\def\diffs1{\diff(S^1)}
\def\psone{\diff^{1,\mathrm{ps}}}
\def\psonezero{\mathrm{Diff}_{+,0}^{1,\mathrm{ps}}}
\def\psoneone{\mathrm{Diff}_{+,1}^{1,\mathrm{ps}}}

\def\vect{{\rm Vect}}
\def\mob{{\rm M\ddot{o}b}}
\def\uMob{{\widetilde{\Mob}}}
\def\vir{{\rm Vir}}

\def\supp{{\rm supp\,}}

\def\psl2r{{\rm PSL}(2,\RR)}
\def\sl2r{{\rm SL}(2,\RR)}
\def\su11{{\rm SU}(1,1)}
\def\2dmob{{\overline{\psl2r}\times\overline{\psl2r}}}
\def\<{\langle}
\def\>{\rangle}

\def\Im{\mathrm{Im}\,}

\def\im{\mathrm{Im}\,}

\def\tremezzi{\mathcal{S}_\frac32(S^1, \RR)}

\def\dom{{\mathscr{D}}}

\DeclareMathOperator{\Mob}{M\ddot ob}

\newcommand{\fin}{\mathrm{fin}}

\newtheorem{theorem}{Theorem}[section]
\newtheorem{lemma}[theorem]{Lemma}
\newtheorem{corollary}[theorem]{Corollary}
\newtheorem{proposition}[theorem]{Proposition}

\theoremstyle{definition}

\theoremstyle{remark}
\newtheorem{remark}[theorem]{Remark}

\numberwithin{equation}{section}

\title{Solitons and nonsmooth diffeomorphisms in conformal nets}
\date{}
\author{
{\bf Simone Del Vecchio}\footnote{Supported by ERC advanced grant 669240 QUEST ``Quantum Algebraic Structures and Models'' and GNAMPA-INDAM}
\\
   Dipartimento di Matematica, Universit\`a di Roma Tor Vergata\\
   Via della Ricerca Scientifica 1, I-00133 Roma, Italy\\
   email: {\tt delvecch@mat.uniroma2.it}\\
\\
{\bf Stefano Iovieno}
\\
   Dipartimento di Matematica, Universit\`a di Roma La Sapienza\\
   Piazzale Aldo Moro 5, I-00185 Roma, Italy\\
   email: {\tt iovieno@mat.uniroma1.it}\\
\\
{\bf Yoh Tanimoto}\footnote{Supported by Programma per giovani ricercatori, anno 2014 ``Rita Levi Montalcini''
of the Italian Ministry of Education, University and Research.}
\\
   Dipartimento di Matematica, Universit\`a di Roma Tor Vergata\\
   Via della Ricerca Scientifica 1, I-00133 Roma, Italy\\
   email: {\tt hoyt@mat.uniroma2.it}
}
\begin{document}

\maketitle

\begin{abstract}
We show that any solitonic representation of a conformal (diffeomorphism covariant) net on $S^1$
has positive energy and
construct an uncountable family of mutually inequivalent solitonic representations of any conformal net,
using nonsmooth diffeomorphisms.
On the loop group nets, we show that these representations induce representations of the subgroup
of loops compactly supported in $S^1\setminus\{-1\}$ which do not extend to the whole loop group.

In the case of the $\uone$-current net, we extend the diffeomorphism covariance to
the Sobolev diffeomorphisms $\D^s(S^1), s > 2$, and show that the positive-energy
vacuum representations of $\diff(S^1)$ with integer central charges extend to $\D^s(S^1)$.
The solitonic representations constructed above for the $\mathrm{U}(1)$-current net and for Virasoro nets
with integral central charge are continuously covariant with respect to
the stabilizer subgroup of $\diff(S^1)$ of $-1$ of the circle.
\end{abstract}

\section{Introduction}
In two-dimensional quantum field theory, solitons appear in the presence of inequivalent vacuum sectors.
In \cite{Froehlich76}, Fr\"olich proposed an operator-algebraic formulation of solitons as superselection
sectors localized in a half-space. Existence of such solitons has been obtained
for a wide class of models \cite{Schlingemann96, Schlingemann98, Mueger99},
and general structural results have been obtained \cite{Fredenhagen93, Rehren98}.
In two-dimensional conformal field theory, the vacuum is unique due to dilation invariance
\cite{Roberts74} and translation-invariant states are not always localized in
half-space \cite{Tanimoto18-1}, yet solitons appear through 
$\alpha$-induction \cite{LR95, BE98, BE99-1}, and the interrelationship between solitons
has led to the operator-algebraic formulation of modular invariant \cite{BE99-2}.
In this way, solitons play a crucial role in the study of conformal field theories.

In the operator-algebraic framework, a conformal field theory is realized as
a family of von Neumann algebras satisfying certain axioms (conformal net),
and the superselection sectors, including solitons, are the equivalence classes of its representations.
When a conformal net has a subnet,
$\alpha$-induction yields solitons for the larger net from a sector of the smaller net.
Note that the Virasoro nets, the conformal nets generated by the stress-energy tensor alone,
do not have any M\"obius covariant subnet \cite{Carpi98}.
Therefore, it is a natural question whether the Virasoro nets admit any nontrivial soliton.
Recently Henriques in \cite{Henriques17-2} proved that the category of solitons $\mathrm{Sol}(\A)$
of a completely rational conformal net $\A$ is a bicommutant category whose Drinfel'd center corresponds
to the category of DHR sectors of $\A$. This fact implies the existence of non-trivial solitons
for all the conformal nets with central charge $c<1$ and $\mu$-index $>1$
(hence including the Virasoro nets with $c<1$),
yet the existence of such solitons is only implicit.
In this paper, we construct for any conformal net a family of concrete, proper, automorphic solitons
(with index $1$) parametrized by $\RR_+$ using nonsmooth diffeomorphisms.

The (smooth) diffeomorphism covariance is the defining property of conformal nets.
This large spacetime symmetry can be further extended to certain nonsmooth diffeomorphisms
depending on their regularity \cite{CW05, Weiner06}.
Here, we show that the representation cannot be extended to some less-smooth
diffeomorphisms, and exploit this to construct proper irreducible solitons
(a similar construction has been implicitly given in $\cite{LX04,KLX05}$ which yielded
non irreducible, type III solitons).
We also show that any soliton has positive energy, this time by exploiting unitarily representable
nonsmooth diffeomorphisms\footnote{Positivity of energy has been proved for finite index representations
without using conformal covariance \cite{BCL98}. Our proof depends on conformal covariance and normality
on half lines, but makes no assumption on the index.}.

This distinction between implementable and non-implementable diffeomorphisms is central in this work.
We say that a nonsmooth diffeomorphism $\gamma$ is unitarily implemented in a conformal net
if there is a unitary operator whose adjoint action realizes $\g$ on the quantum observables.
When a diffeomorphism is not implementable, its action may give rise to a new sector.
In any conformal net,
the Sobolev-class diffeomorphisms $\D^s(S^1)$ with $s>3$ are implementable \cite{CDIT18+}.
Here we show that, using the Tomita-Takesaki modular theory, nonsmooth diffeomorphism
which have discontinuous derivatives are not implementable, although they are implementable
when restricted to local algebras.
In this way, we obtain an uncountable family of inequivalent solitons for any conformal net.
On the other hand, the implementability of many nonsmooth diffeomorphisms is
inherited by sectors, and indeed positivity of energy in any soliton is proved in this way.

As an application, we construct irreducible unitary projective positive-energy representations
of the subgroup $\Lambda G$ of the loop group $LG$, consisting of loops with support not containing the point $-1$
which do not extend to $LG$. The existence of such representations
was marked as an open problem in \cite[P.174, Remark]{PS86}.
Similar representations are constructed for the group $B_0$ of diffeomorphisms of $S^1$ preserving the point $-1$.
These results can be seen as an application of the modular theory of von Neumann algebras to
the representation theory of infinite-dimensional groups.

We also pursue the question of which nonsmooth diffeomorphisms are implementable.
We take the $\mathrm{U}(1)$-current net (the derivative of the massless free field, or the Heisenberg algebra)
which has the central charge $c=1$, and show that it is covariant with respect
to Sobolev-class diffeomorphisms $\D^s(S^1), s>2$. This is done by the Shale(-Steinspring) criterion
of unitary implementation \cite{Shale62}, and improve the implementation with $s>3$ for general conformal nets \cite{CDIT18+}.
This implies that some unitary representations of $\diff(S^1)$ with integer $c$ can be
extended to $\D^s(S^1)$-diffeomorphisms with $s>2$.
This is to our knowledge the largest group which is implementable for some $c$,
and no other implementable diffeomorphism is known.
As a consequence, the solitons we construct in Section \ref{nonsmoothsoliton}
are continuously $B_0$-covariant in this case.

This paper is organized as follows: in Section \ref{preliminaries} we recall the notions of conformal net and its representation theory
together with some facts about the diffeomorphisms groups and loop groups.
In Section \ref{positivity} we prove that every soliton is translation covariant with positive energy.
The covariance can be further extended, but we do not know the continuity.
In Section \ref{nonsmoothsoliton} we construct a family of proper solitons arising from nonsmooth
diffeomorphisms. It is shown in Section \ref{DHR} that a soliton which is locally $\mob$-covariant
(not just covariant with respect to the universal covering $\uMob$) extends to a DHR representation.
Section \ref{examples} is dedicated to concrete examples: we use the results in Section \ref{general}
to prove that there exist irreducible positive energy representations of $B_0$ and $\Lambda G$ which do not
extend to $\diff(S^1)$ and $LG$, respectively.
Furthermore, in Section \ref{uone} we show that the $U(1)$-current net and the Virasoro nets
with positive integer central charge are $\D^s(S^1)$-covariant, $s>2$.
In Section \ref{outlook} we summarize open problems.
In Appendix, we prove that any conformal net is covariant with respect to
piecewise smooth $C^1$-diffeomorphisms.

\section{Preliminaries}\label{preliminaries}
\subsection{Conformal nets}
Let $\mathcal{I}$ be the set of nonempty, non-dense, connected open intervals of the unit circle $S^{1}$.
For $I \in \I$, $I^{c}$ denotes the interior of the complement of the interval $I\in\mathcal{I}$,
namely $I^{c}=(S^{1}\setminus I)^{\circ}$.
The M\"obius group $\mob = \psl2r$ acts on $S^1$ by linear fractional transformations (see Section \ref{moebius}).

A {\bf M\"obius covariant net on $S^{1}$} $(\A, U, \Omega)$ consists of
a family $\left\{\mathcal{A}(I), I\in\mathcal{I}\right\}$
of von Neumann algebras acting on a separable complex Hilbert space $\mathcal{H}$,
a strongly continuous unitary representation $U$ of $\mob$ and a ``vacuum'' vector $\Omega \in \H$
satisfying the following properties:
\begin{enumerate}[{(CN}1{)}]
\item {\bf Isotony}: if $I_{1}\subset I_{2}$, $I_{1},I_{2}\in \mathcal{I}$, then
$\mathcal{A}(I_{1})\subset\mathcal{A}(I_{2})$.
\item {\bf Locality}: if $I_{1}\cap I_{2}=\emptyset$, $I_{1},I_{2}\in \mathcal{I}$, then
$\mathcal{A}(I_{1})\subset\mathcal{A}(I_{2})^{\prime}$.
\item {\bf M\"obius covariance}: for $g\in \mob$, $I\in\mathcal{I}$,
$\Ad U(g)(\mathcal{A}(I)) = \mathcal{A}(gI)$.
\item {\bf Positivity of energy}: the representation $U$ has positive energy, i.e.\! the conformal Hamiltonian $L_{0}$ (the generator of rotations) has non-negative spectrum.
\item {\bf Vacuum}: $\Omega$ is the unique (up to a scalar) vector such that $U(g)\Omega=\Omega$ for $g\in\mob$,
and $\Omega$ is cyclic for $\bigvee_{I\in\mathcal{I}}\mathcal{A}(I)$.
\setcounter{CN}{\value{enumi}}
\end{enumerate}
Positivity of energy is actually equivalent to positivity of the generator of
translations \cite[Proposition 1]{Koester02}, see Section \ref{moebius}.
With these assumptions, the following automatically hold,
see \cite[Lemma 2.9, Theorem 2.19(ii)]{GF93}\cite[Section 3]{FJ96} and the arguments of \cite[Theorem 1.2.6]{Baumgaertel95}:
\begin{enumerate}[{(CN}1{)}]
\setcounter{enumi}{\value{CN}}
\item {\bf Haag duality}: for every $I\in\mathcal{I}$, $\mathcal{A}(I')=\mathcal{A}(I)'$ where $\mathcal{A}(I)'$ is the commutant of $\mathcal{A}(I)$.
\item {\bf Additivity}: if $I, I_{\alpha}\in\mathcal{I}$ and $I \subset \bigcup_\a I_\a$,
then $\mathcal{A}(I)\subset \bigvee_{\alpha}\mathcal{A}(I_{\alpha})$.
\item\label{bisognano} {\bf Bisognano-Wichmann property}:
if $I_{(0,\pi)}=\{z\in S^1: \text{Im}(z)>0\}\in\mathcal{I}$ and $\Delta_{I_{(0,\pi)}}$ is the modular operator associated to $\mathcal{A}(I_{(0,\pi)})$ and $\Omega$ then
\begin{align*}
\Delta^{it}_{I_{(0,\pi)}}=U(\delta(-2\pi t)),
\end{align*}
where $\delta$ is the one parameter group of dilations (as a subgroup of $\mob$ through the Cayley transform, see Section \ref{moebius}).
\item\label{cn:irreducibility} {\bf Irreducibility}: each $\A(I)$ is a type III factor and
$\bigvee_{I\in{\overline{\I}_\RR}} \mathcal{A}(I)=\B(\H)$, where ${\overline{\I}_\RR}$ is the set of intervals
not containing the point $-1$, see Section \ref{representations}.
\setcounter{CN}{\value{enumi}}
\end{enumerate}

By a {\bf conformal net} (or diffeomorphism covariant net)
we shall mean a M\"obius covariant net which satisfies the following:
\begin{enumerate}[{(CN}1{)}]
\setcounter{enumi}{\value{CN}}
\item $U$ extends to a projective unitary representation of $\diff(S^1)$ on $\mathcal{H}$
such that for all $I\in\mathcal{I}$ we have
\[
 \Ad U(\g)(\mathcal{A}(I)) = \mathcal{A}(\g I), \quad \g\in\diff(S^1),
\]
and 
\begin{align}\label{eq:diffcov2}
\Ad U(\g)(x) = x, \quad x\in \A(I),\quad \supp \g\in I^\prime,
\end{align}
where $\supp \g$ is the closure of the complement of the set of $z \in S^1$ such that $\g(z)=z$.
\setcounter{CN}{\value{enumi}}
\end{enumerate}
In a conformal net, the following is automatic \cite{MTW18}:
\begin{enumerate}[{(CN}1{)}]
\setcounter{enumi}{\value{CN}}
\item\label{split} {\bf The split property}:
if $\overline{I} \subset \tilde I$ where $\overline I$ is the closure of $I$,
then there is a type I factor $\mathcal{R}_{I,\tilde I}$ such that
$\A(I) \subset \mathcal{R}_{I,\tilde I} \subset \A(\tilde I)$.
\setcounter{CN}{\value{enumi}}
\end{enumerate}

\subsection{Representations of conformal nets}\label{representations}
\paragraph{DHR representations.}
A {\bf DHR (Doplicher-Haag-Roberts) representation} $\rho$ of a conformal net $\mathcal{A}$ is a family of maps
$\{\rho_I\}_{I\in\I}$ where $\rho_I$ is a normal ($\sigma$-weakly continuous) representation of the von Neumann algebra $\mathcal{A}(I)$
on a fixed Hilbert space $\mathcal{H}_{\rho}$ with the compatibility property
$\rho_{I_2}|_{\mathcal{A}(I_1)} = \rho_{I_1}, I_1\subset I_2$.

We say that two representations $\rho_1, \rho_2$ are unitarily equivalent if there exists an intertwining unitary operator
$U$ from $\mathcal{H}_{\rho_1}$ and $\mathcal{H}_{\rho_2}$, i.e. $U\rho_{1,I}(x)=\rho_{2,I}(x)U$ for every $x\in\mathcal{A}(I)$
and $I\in\mathcal{I}$.
A representation $\rho$ is said to be {\bf irreducible} if $\bigvee_{I\in \I}\rho_I(\mathcal{A}(I)) = \B(\H_\rho)$.
The collection of identity maps $\rho_0 = \{\rho_{0,I}\}$
where $\rho_{0,I}(x) = x, x \in \A(I)$ is called the {\bf vacuum representation}.

A group $\mathbf{G}$ may act on $S^1$.
In this paper, $\mathbf{G}$ will be either $\mob, \diff(S^1)$ or some groups which contain $\diff(S^1)$
such that $U$ can be extended to them (see Sections \ref{nonsmoothsoliton} and \ref{uone}).
A DHR representation $\rho$ is said to be {\bf $G$-covariant} if there exists a unitary projective
representation $U_\rho$ of $\mathbf{G}$ such that
\begin{align*}
\Ad U_\rho(\gamma)(\rho_I(x)) = \rho_{\gamma I}(\Ad U(\gamma)(x)),
\end{align*}
for all $\gamma\in \mathbf{G}$.
If $\mathbf{G}$ is a topological group and acts on $S^1$ continuously,
we often require that $U$ is continuous in the strong operator topology.

\paragraph{Solitons.}
Let ${\overline{\I}_\RR}$ be the set of open, non-empty, connected subsets of the real line $\mathbb{R}$,
identified with $S^1\setminus\{-1\} \subset \CC$ via Cayley transform (see Section \ref{moebius}).
Namely, ${\overline{\I}_\RR}$ is the family of bounded open intervals and open half-lines of $\mathbb{R}$.

A {\bf soliton} (or solitonic representation to be precise) $\sigma$ of a conformal net $\mathcal{A}$
is a family of maps $\{\sigma_I\}_{I\in{\overline{\I}_\RR}}$
where $\sigma_I$ is a normal representation of the von Neumann algebra $\mathcal{A}(I)$
on a fixed Hilbert space $\mathcal{H}_{\sigma}$ with the compatibility property
$\rho_{I_2}|_{\mathcal{A}(I_1)} = \rho_{I_1}, I_1\subset I_2$.
We say that the soliton $\sigma$ is {\bf proper} if there is no DHR representation of
the conformal net $\mathcal{A}$ which agrees with $\sigma$ when restricted to the family of intervals ${\overline{\I}_\RR}$.

Let $\mathbf{G}$ again be a group acting on $S^1$, and $\mathbf{G}_0 \subset \mathbf{G}$ be a subgroup whose
elements preserve the point $-1 \in S^1$.
A soliton $\sigma$ of $\mathcal{A}$
is {\bf $\mathbf{G}_0$-covariant} if there is a unitary projective representation $U_\sigma$ of $\mathbf{G}_0$
such that $\Ad U_\sigma(\gamma)(\sigma_I(x)) = \sigma_{\gamma I}(\Ad U(\gamma)(x))$
with $x\in\mathcal{A}(I)$. 
Similarly, a soliton $\sigma$ is {\bf locally $\mathbf{G}$-covariant}\footnote{In
\cite[Definition 2.2]{CHKLX15}, the notion ``$\mathbf{G}$-covariance'' is
defined through a unitary projective representation of the universal covering group $\widetilde{\mathbf{G}}$.
Differently from this, we distinguish explicitly projective representations of $\mathbf{G}$ and of
$\widetilde{\mathbf{G}}$ and add ``locally''.}
if there is a unitary projective representation $U_\sigma$ of $\mathbf{G}$
such that
for a neighborhood $\U$ of the unit element of $\mathbf{G}$ and $I\in \overline{I}_\RR$
such that $\gamma I \subset \RR$ for $\g \in \U$, it holds that 
$\Ad U_\sigma(\gamma)(\sigma_I(x)) = \sigma_{\gamma I}(\Ad U(\gamma)(x)), x\in \A(I), \g \in \U$.

Consider the case where $\mathbf{G}_0$ includes the translation group of $\RR$.
We say that a soliton $\sigma$ has {\bf positive energy} if the unitary representation
$U_\sigma$ above can be chosen in such a way that the restriction to the one-parameter subgroup
of translations is continuous in the strong operator topology and has a positive generator.
Note that, if $\sigma$ is not irreducible, $U_\sigma$ which implements covariance is not unique,
and other implementations may fail to have positive generator.

With $\RR_\pm$ considered as intervals in ${\overline{\I}_\RR}$,
we define the {\bf index} of $\sigma$ as the Jones index
of the inclusion $\sigma(\A(\RR_+)) \subset \sigma(\A(\RR_-))'$.
This is a natural generalization of the index of DHR representations.

\subsection{The spacetime symmetry groups}
\subsubsection{The M\"obius group}\label{moebius}
The group $\mathrm{SL}(2,\mathbb{R})$ of $2\times 2$ real matrices with determinant one acts on the compactified real line $\mathbb{R}\cup\left\{\infty\right\}$ by linear fractional transformations:
\begin{align*}
g:t\rightarrow gt\coloneqq\frac{at+b}{ct+d}\hspace{5mm}\text{for}\hspace{3mm}g=\left(
\begin{matrix}
a & b\\
c & d
\end{matrix}
\right)\in \mathrm{SL}(2,\mathbb{R}).
\end{align*}
The kernel of this action is $\left\{\pm \1\right\}$.
By identifying the compactified real line $\mathbb{R}\cup\{\infty\}$ with the circle $S^{1}$
via Cayley transform
\begin{align}\label{eq:cayley}
C:S^1\setminus\{-1\}\rightarrow \mathbb{R},\qquad
z\mapsto i\frac{1-z}{1+z},
\end{align}
with inverse
\begin{align}\label{eq:cayleyinv}
C^{-1}:\mathbb{R}\rightarrow S^1\setminus\{-1\}, \qquad
t\mapsto \frac{1+it}{1-it},
\end{align}
the group $\psl2r:=\mathrm{SL}(2,\mathbb{R})/\left\{\pm \1\right\}$ can be identified with
a subgroup of diffeomorphisms of the circle $S^{1}$, the M\"obius group $\mob$.

The following are important subgroups of $\psl2r$:
\begin{align*}
R(\theta)=\left(
\begin{matrix}
\cos(\theta/ 2) & \sin(\theta/ 2)\\
-\sin(\theta/ 2) & \cos(\theta/ 2) 
\end{matrix}
\right), \hspace{5mm}
\delta(t)=
\left(
\begin{matrix}
e^{t/ 2} & 0\\
0 & e^{-t/ 2} 
\end{matrix}
\right),\hspace{5mm}
\tau(t)=
\left(
\begin{matrix}
1 & t\\
0 & 1
\end{matrix}
\right),
\end{align*}
and they are called the rotation, dilation and translation subgroup, respectively,
and act in the following way:
\begin{align*}
R(\theta)z&=e^{i\theta}z\hspace{5mm}\text{on}\hspace{1.5mm}S^{1},\\
\delta(t)s&=e^{t}s\hspace{5mm}\text{on}\hspace{1.5mm}\mathbb{R},\\
\tau(t)s&=s+t\hspace{5mm}\text{on}\hspace{1.5mm}\mathbb{R}.
\end{align*}
The expressions of $\delta(t)$ and $\tau(t)$ in the circle picture are available in
\cite[(A.6)]{WeinerThesis}.

\subsubsection{The (smooth) diffeomorphism groups}\label{smooth}
\paragraph{The Lie group.}
Let us denote by $\diff(S^1)$ the group of orientation preserving, smooth diffeomorphisms
of the circle $S^1\coloneqq \lbrace z\in\CC :\vert z\vert=1\rbrace$ and $\vect(S^1)$ denote
the set of smooth vector fields on $S^1.$
The group $\diff(S^1)$  is an infinite-dimensional Lie group whose Lie algebra is identified
with the real topological vector space $\vect(S^1)$ of smooth vector fields on $S^1$ with
$C^\infty$-topology \cite{Milnor84}.
The exponential map $\Exp:\vect(S^1)\rightarrow\diff(S^1)$ maps $tf\in\vect(S^1)$
to the one-parameter group $\Exp(tf)\in\diff(S^1)$ of diffeomorphisms of $S^1$ satisfying
the ordinary differential equation
\begin{align*}
\frac{d\Exp(tf)(z)}{dt}=f(\Exp(tf)(z)), \qquad \Exp(0)(z) = z.
\end{align*}
We identify $\vect(S^1)$ with $C^\infty(S^1,\mathbb{R})$ and
for $f\in C^{\infty}(S^1,\mathbb{R})$ we denote by $f^\prime$ the derivative of $f$ with respect to the angle
$\theta$:
\[
  f^\prime(z)=\frac{d}{d\theta}f(e^{i\theta})\bigg\rvert_{e^{i \theta}=z}.
\]
We consider a diffeomorphism $\gamma\in\diff(S^1)$ as a map from $S^1 \subset \CC$
in $S^1$. With this convention, its action on $f\in \vect(S^1)$ is 
\begin{align*}
(\gamma_* f)(e^{i\theta})=
-ie^{-i\theta}\frac{d}{d\varphi}\gamma(e^{i\varphi})\bigg\rvert_{e^{i\varphi} = \gamma^{-1}(e^{i\theta})}f(\gamma^{-1}(e^{i\theta})). 
\end{align*}

\paragraph{The Lie algebra.}
The space $\vect(S^1)$ is endowed with the Lie algebra structure with
the Lie bracket given by
\[
 [f,g]=f^{\prime}g-f g^{\prime}.
\]
As a Lie algebra, $\vect(S^1)$ admits the Gelfand--Fuchs 2-cocycle 
\begin{align*}
\omega (f,g)=\frac{1}{48\pi}\int_{S^1}(f(e^{i\theta})g^{\prime\prime\prime}(e^{i\theta})-f^{\prime\prime\prime}(e^{i\theta})g(e^{i\theta}))d\theta.
\end{align*}
The Virasoro algebra $\vir$ is the central extension of the complexification of
the algebra generated by the trigonometric polynomials in $\vect(S^1)$ defined by the 2-cocycle $\omega$.
It can be explicitly described as the complex Lie algebra generated by $L_n$, $n\in\mathbb{Z}$,
and the central element $\kappa$, with brackets
\[
 [L_n,L_m]=(n-m)L_{n+m}+\delta_{n+m,0}\frac{n^3-n}{12}\kappa.
\]

\paragraph{Positive-energy representations.}
Consider a representation $\pi_V:\vir\rightarrow\End(V)$ of $\vir$ on a complex vector space $V$ endowed with a
non-degenerate, positive-definite scalar product $\langle\cdot,\cdot\rangle$. We call $\pi_V$ a
{\bf unitary positive energy representation} if the following hold:
\begin{itemize}
\item Unitarity: $\langle v,\pi_V(L_n)w\rangle=\langle \pi_V(L_{-n})v,w\rangle$
for every $v,w\in V$ and $n\in\ZZ$;
\item Positivity of the energy: $V=\bigoplus_{\lambda\in\RR_+\cup\lbrace 0\rbrace}V_{\lambda}$, where $V_{\lambda}\coloneqq \ker(\pi_V(L_0)-\lambda\1_V)$. The lowest eigenvalue of $\pi_V(L_0)$
is called lowest weight;
\item  Central charge: $\pi_V(\kappa)=c\1_V$;
\end{itemize}
There exists an irreducible unitary positive energy representation $\pi_{c,h}$ with central charge $c$
and lowest weight $h$ if and only if $c\ge 1$ and $h\ge 0$ (continuous series representation)
or $(c,h)=(c(m),h_{p,q}(m))$, where $c(m)=1-\frac{6}{(m+2)(m+3)}$, $h_{p,q}(m)=\frac{(p(m+1)-qm)^2-1}{4m(m+1)}$,
$m=3,4,\cdots$, $p=1,2,\cdots,m-1$, $q=1,2,\cdots,p$, (discrete series representation) \cite{KR87, DMS97}.
In this case the representation space $V$ is denoted by $\H^\fin(c,h)$.
We denote by $\mathcal{H}(c,h)$ the Hilbert space completion of the vector space $\H^\fin(c,h)$ associated
with the unique irreducible unitary positive energy representation of $\vir$
with central charge $c$ and lowest weight $h$.

\paragraph{The stress-energy tensor.}
In a (possibly infinite) direct sum representation $\pi_\H$ of $\pi_{c,h_j}$ with the same central charge $c$,
the conformal Hamiltonian $L_0$ is diagonalized,
and on the linear span of its eigenvectors $\mathcal{H}^{\fin}$ (the space of finite energy vectors),
the Virasoro algebra acts algebraically as unbounded operators.
With a slight abuse of notation, we denote by $L_n$ the elements of $\vir$ represented in $\mathcal{H}$.

For a smooth complex-valued function $f$ on $S^1$ with finitely many non-zero Fourier coefficients,
the (chiral) stress-energy tensor associated with $f$ is the operator
\[
 T(f)=\sum_{n\in\mathbb{Z}}L_n \hat{f}_n
\]
acting on $\mathcal{H}$, where
\[
 \hat{f}_n=\int_0^{2\pi}\frac{d\theta}{2\pi}e^{-in\theta}f(e^{i\theta})
\]
The stress-energy tensor $T$ can be extended to certain nonsmooth real functions
$f\in \tremezzi$
by the linear energy bounds, yielding a self-adjoint unbounded operator $T(f)$.
We will review these results in detail in Section \ref{piecewise}.

Such a representation integrates to a projective unitary representation of
the universal covering group $\widetilde{\diff(S^1)}$ of $\diff(S^1)$,
namely, there is a projective unitary representation $U$ of $\widetilde{\diff(S^1)}$
such that $U(\Exp(f)) = e^{iT(f)}$ up to a scalar\footnote{This scalar cannot be made trivial \cite[(5.11)]{FH05}.}.
If the $h_j$'s appearing in the direct summand differ from each other by an integer,
then $U$ reduces to a projective representation of $\diff(S^1)$ ($2\pi$-rotation is a scalar).

The stress-energy tensor $T$ satisfies the following covariance
\cite[Proposition 5.1, Proposition 3.1]{FH05}:
\begin{proposition}\label{pr:covariance}
The stress-energy tensor $T$ on $\mathcal{H}$ and its integration $U$ as above satisfy
\begin{align*}
 \Ad U(\gamma)(T(f)) &= T(\mathring\gamma_*({f}))
                        + \beta(\mathring\gamma, f) \\
 \beta(\mathring\gamma, f) &\coloneqq \frac{c}{24\pi}\int^{2\pi}_0\{\mathring\gamma,z\}|_{z = e^{i\theta}}f(e^{i\theta})e^{i2\theta}d\theta
\end{align*}
on vectors in $\mathcal{H}^{\mathrm{fin}}$, for $f\in\vect(S^1)$ and $\gamma\in\widetilde{\diff(S^1)}$,
where $\mathring\gamma \in \diff(S^1)$ is the image of $\gamma$ under the covering map.
Furthermore, the commutation relations
\[
 i[T(g),T(f)]=T(g^\prime f-f^\prime g)+ c \omega(g,f),
\]
hold for arbitrary $f,g\in C^\infty (S^1)$, on vectors $\psi\in \mathcal{H}^{\mathrm{fin}}$,
where
\[
 \{\mathring{\gamma},z\}=\frac{\frac{d^3}{dz^3}\mathring{\gamma}(z)}{\frac{d}{dz}\mathring{\gamma}(z)}
            -\frac{3}{2}\left(\frac{\frac{d^2}{dz^2}\mathring{\gamma}(z)}{\frac{d}{dz}\mathring{\gamma}(z)}\right)^2 
\]
is the Schwarzian derivative of $\mathring{\gamma}$ and $\frac{d}{dz}\mathring{\gamma}(z)=-i\bar{z}\frac{d}{d\theta}\mathring{\gamma}(e^{i\theta})\bigg\rvert_{e^{i\theta}=z}$.
\end{proposition}

If we consider the Cayley transform \eqref{eq:cayley}\eqref{eq:cayleyinv},
a vector field $f\in \vect(S^1)$ in real line coordinates is given by
\[
 C_*(f)(t)=\frac{2}{(1+t^2)}f(C^{-1}(t)).
\]
With the Schwarz class functions $\mathscr{S}(\mathbb{R})$,
the stress energy tensor satisfies the following quantum-energy inequalities \cite[Theorem 4.1]{FH05}.
\begin{theorem}\label{th:qei}
Let $f\in\vect(S^1)$ with $C_*(f)\in \mathscr{S}(\mathbb{R})$ and $C_*(f)(t)\geq0$ $\forall t\in\mathbb{R}$.
For $\psi\in \dom(L_0)$, it holds that
\[
 (\psi,T(f)\psi)\geq - \frac{c}{12\pi}\int_{\mathbb{R}}\left(\frac{d}{dt}\sqrt{C_*(f)(t)}\right)^2 dt,
\]
where the derivative is given by
\begin{align*}
\frac{d}{dt}\sqrt{C_*(f)(t)}=
\begin{cases}   (\frac{d}{dt}C_*(f)(t))/(2\sqrt{C_*(f)(t)}) &\text{ if } C_*(f)(t)\neq 0\\
0 & \text{ if } C_*(f)(t)= 0.
\end{cases}
\end{align*}
\end{theorem}

\paragraph{Translations}
The generator of translations is by definition
$\gt(t) \coloneqq \frac{\partial}{\partial s}\left(\tau(t)s\right)\bigg\vert_{s=0}=1$. The corresponding field on $S^1$ is
\begin{align*}
\gt(e^{i\theta}) = 1+\cos \theta.
\end{align*}

\paragraph{Subgroups $B_n$}
The subgroup $B_0 \subset \diff(S^1)$ of the stabilizers of $-1$
is a Lie subgroup with Lie algebra given by those vector fields $f\in\vect(S^1)$
such that $f(-1)=0$. The dilation and translation subgroups of $\diff(S^1)$ are in $B_0$.
Similarly, $B_n$ is the subgroup of $B_0$ whose element have vanishing
$1$st, $2$nd, $\cdots$, $n$-th derivatives.
The group $B_1$ still contains translations, but not dilations.
These groups have the natural $C^\infty$ topology, but we often treat them without topology.

\subsubsection{Piecewise smooth diffeomorphisms}\label{piecewise}
Let $\psone(S^1)$ be the group of piecewise smooth $C^1$-diffeomorphisms of $S^1$,
namely, $\g \in \psone(S^1)$ is a $C^1$-diffeomorphism and $S^1$ can be decomposed into a finitely many closed
intervals (with a possibly common end point) on each of which $\g$ is smooth and has the derivatives in all orders
at the end points.
It has been known that some elements of $\psone(S^1)$ can be implemented in a conformal net \cite{CW05, Weiner06}.
Let us recall these elements.

For a real-valued continuous function $f$ of the circle, set
\[
\Vert f\Vert_{\frac{3}{2}}\coloneqq \sum_{n\in\mathbb{Z}}\vert{\hat{f}}_n\vert(1+|n|^{\frac{3}{2}}),
\]
where $\hat{f}_n\coloneqq \frac{1}{2\pi}\int_0^{2\pi}e^{-in\theta}f(e^{i\theta})d\theta$
is the $n$-th Fourier coefficient\footnote{This should be distinguished from a sequence of functions $f_n$.} of $f$. 
We denote with $\mathcal{S}_{\frac{3}{2}}(S^1,\mathbb{R})$ the class of functions $f\in L^1(S^1,\RR)$ such that $\Vert f\Vert_{\frac{3}{2}}$ is finite endowed with the topology induced by the norm $\Vert\cdot\Vert_{\frac{3}{2}}$.
By \cite[Lemma 2.2]{Weiner06},
if $f$ is piecewise smooth and and $C^1$ on the whole $S^1$, then
$f \in \mathcal{S}_{\frac{3}{2}}(S^1,\mathbb{R})$.

Let $T$ be the stress-energy tensor on $\H = \bigoplus_j \H(c,h_j)$.
In \cite[Proposition 4.2, Theorem 4.4, Proposition 4.5]{CW05} it has been shown that,
if $f \in \tremezzi$, then the operator $T(f)=\sum_{n\in\mathbb{Z}}L_n f_n$ on the domain $\mathcal{H}^{\fin}$
is convergent and essentially self-adjoint on any core of $L_0$.
In addition, if $\|f-f_n\|_{\frac{3}{2}} \to 0$ for $f, f_n \in \tremezzi$, then
$T(f_n)\rightarrow T(f)$ in the strong resolvent sense.

From these results, it follows that certain piecewise smooth $C^1$-diffeomorphisms
of the form $\Exp(f)$ are implemented in a conformal net. Actually, we prove in Appendix \ref{c1pws}
that any conformal net is $\psone(S^1)$-covariant. We do not consider topology on $\psone(S^1)$.
It follows that any soliton is $\psoneone(S^1)$-covariant,
where $\psoneone(S^1) \coloneqq \{\gamma \in \psone(S^1): \gamma(-1) = -1, \gamma'(-1) = 1\}$
(Theorem \ref{th:positivitysol}).
The solitons we construct in Section \ref{nonsmoothsoliton} are $\psonezero(S^1)$-covariant
where $\psonezero(S^1) \coloneqq \{\gamma \in \psone(S^1): \gamma(-1) = -1\}$.

\subsubsection{The Groups of Sobolev-class diffeomorphisms}
For $s$ real, the Sobolev spaces $H^s(S^1)$ are defined by
\begin{align*}
 H^s(S^1)\coloneqq \{f\in L^2(S^1): \|f\|_{H^s} < \infty\}, \text{ where }
 \|f\|_{H^s} := \left(\sum_{k\in\ZZ} (1+k^2)^s|\hat f_k|^2\right)^\frac12.
\end{align*}
If $s > \frac32$, then the set $\D^s(S^1)$ of Sobolev-class diffeomorphisms
\begin{align*}
 \D^s(S^1)\coloneqq \{\g \in \diff^1(S^1): \tilde \g - \i \in H^s\},
\end{align*} 
where $\tilde \g$ is a lift of $\g$ to $\RR$, is a topological group \cite[Theorem B.2]{IKT13}
(see also \cite[Lemma 2.5]{CDIT18+}). We also have the following continuity of the action on $H^s(S^1)$ \cite[Theorem B.2]{IKT13}.
\begin{lemma}\label{lm:sobolevcomp}
For $s>3/2$, the map $(f,\g)\mapsto f\circ\g$ from $H^s(S^1)\times\D^s(S^1)$ into $H^s(S^1)$ is continuous.
\end{lemma}

\section{General results on solitons}\label{general}
Throughout this Section, $(\A, U, \Omega)$ is a conformal net in the sense of Section \ref{preliminaries}.

It has been observed \cite[Section 3.3.1]{Henriques17-1} that, by \cite{Weiner06},
any soliton can be made translation covariant.
In the next Section we show further that it has always positive energy,
proving \cite[Conjecture 32]{Henriques17-1}.
We suspect that the converse implication \cite[Conjecture 34]{Henriques17-1} could be negative,
cf.\! \cite{Tanimoto18-1, Tanimoto11}.
In addition, we present a general scheme to construct solitons for any conformal net.
\subsection{Positivity of energy}\label{positivity}

Let us first observe that $\Exp(tg)$ makes sense if $g$ is $C^1$, because then
the existence and uniqueness of solution of the ordinary differential equation
are assured \cite[Chapter 1, Theorem 2.3]{CL55}. We need some preparatory results
on representations of these elements.

\begin{lemma}\label{lm:smoothaction}
 Let $g \in C^\infty(S^1, \RR)$ and
 $f$ be a real piecewise smooth and $C^1$-function on $S^1$.
 Then it holds that
 \begin{align*}
  \Ad e^{iT(g)}(T(f)) = T(\Exp (g)_*(f)) + \beta(\Exp (g),f)
 \end{align*}
\end{lemma}
\begin{proof}
 Let us fix $s$ such that $2 < s < \frac52$. Note that $f\in H^s(S^1)$.
 Indeed, $f''$ is defined except a finite number of points and of bounded variation, and by the proof of \cite[Lemma 2.2]{Weiner06},
 we have $|k^2 \hat f_k| \le \left|\frac{\mathrm{Var}(f'')}{k}\right|$,
 where $\mathrm{Var}(f'')$ is the variation of $f''$, see \cite[Theorem 4.5]{Katznelson04}.
 From this it is immediate that $|k|^{2s} |\hat f_k|^2 \le \left|\frac{\mathrm{Var}(f'')^2}{k^{6-2s}}\right|$
 and the right-hand side is summable in $k$ as $6-2s > 1$, hence $f \in H^s(S^1)$.
 
 Next, let us observe that $H^s(S^1) \subset \mathcal{S}_\frac32(S^1)$.
 Indeed,
 \begin{align*}
  \sum_k (1+|k|)^\frac32 |\hat f_k| &\le \sum_k (1+|k|)^s |\hat f_k| \cdot (1+|k|)^{\frac32 - s}
  \le 2\sum_k (1+|k|^2)^\frac s2 |\hat f_k| \cdot (1+|k|)^{\frac32 - s}
 \end{align*}
 and the right-hand side can be seen as a scalar product of two $\ell^2(\ZZ)$ sequences
 (because $s > 2$),  hence it holds that $\|f\|_\frac32 \le \mathrm{Const.} \|f\|_{H^s}$,
 where the constant depends on $s$ but not on $f$.

 We know that there is a sequence $\{f_n\}\subset C^\infty(S^1, \RR), \|f-f_n\|_{H^s} \to 0$.
 For $f_n \in C^\infty(S^1, \RR)$, we have by Proposition \ref{pr:covariance}
 \begin{align}\label{eq:smoothactionsequence}
  \Ad e^{iT(g)}(e^{iT(f_n)}) = e^{i(T(\Exp (g)_*(f_n)) + \beta(\Exp (g),f_n))}.
 \end{align}
 By the above observations, we have $f_n \to f$ in $\tremezzi$.
 By \cite[Lemma 2.5(a)]{CDIT18+} (see also \cite[Lemma B.2]{IKT13}),
 $f\mapsto \Exp (g)_*(f)$ is continuous in $H^s(S^1)$, hence $\Exp(g)_*(f_n) \to \Exp(g)_*(f)$ in $\mathcal{S}_\frac32(S^1)$.
 By \cite[Proposition 4.5]{CW05}, $T(\Exp (g)_*(f_n)) \to T(\Exp (g)_*(f))$
 in the strong resolvent sense, and it is also clear that
 $\beta (\Exp (g),f_n) \to \beta(\Exp (g),f)$.
 Therefore, by taking the limit of \eqref{eq:smoothactionsequence},
 we obtain the claim.
\end{proof}
\begin{remark}
 If $f \in C^1$ and not $C^2$, then $f\notin H^s(S^1)$ for $s > \frac52$ since
 with such $s$ it holds that $H^s(S^1) \subset C^2(S^1)$ by the Sobolev-Morrey embedding.

In \cite[Proposition 2.3]{Weiner06}, it is claimed that the same conclusion holds for $f \in \mathcal{S}_\frac32(S^1)$,
but a proof of the convergence $f_n \circ \g \to f \circ \g$ in $\|\cdot\|_\frac32$ is missing.
Yet, the main results of the paper remain valid because one needs only the conclusion for $f$ which is
piecewise smooth and $C^1$.
\end{remark}

\begin{lemma}\label{lm:nonsmoothaction}
 Let $g,f \in C^\infty(S^1, \RR)$ and $g(-1) = g' (-1) = f(-1) = f'(-1) = 0$ and compactly supported.
 Let $I_\pm$ be disjoint intervals in $S^1$ one of whose boundary points is $-1$ (see Figure \ref{fig:intervals}.
 Let $f=f_- + f_+, f_\pm \in \mathcal{S}_{\frac32}(S^1)$
 be the decomposition of $f$ into two pieces cut at the point $-1$
 (which is possible by \cite[Lemma 2.2]{Weiner06}),
 and similarly introduce $g=g_- + g_+, g_\pm \in \mathcal{S}_{\frac32}(S^1)$,
 and assume that $\supp f_\pm, \supp g_\pm \subset I_\pm$.
 \begin{figure}[ht]
\centering
\begin{tikzpicture}[line cap=round,line join=round,>=triangle 45,x=1.0cm,y=1.0cm, scale = 0.7]
\clip(-5,-3) rectangle (5,3);
\draw [line width=1pt] (0,0) circle (2.5cm);
\draw [line width=3pt] (-2.5,0) arc (180:90:2.5) ;
\draw [line width=3pt] (-2.5,0) arc (180:240:2.5) ;

\draw [line width=1pt] (0,2.5) arc (0:30:0.5) ;
\draw [line width=1pt] (0,2.5) arc (0:-30:0.5) ;

\draw [line width=1pt] (-2.5,0) arc (90:120:0.5) ;
\draw [line width=1pt] (-2.5,0) arc (90:60:0.5) ;

\draw [line width=1pt] (-2.5,0) arc (270:300:0.5) ;
\draw [line width=1pt] (-2.5,0) arc (270:240:0.5) ;

\draw [line width=1pt] (-1.25,-2.16506) arc (330:360:0.5) ;
\draw [line width=1pt] (-1.25,-2.16506) arc (330:300:0.5) ;

\draw (-2.7,0) node[anchor=east] {$-1$};
\draw (-2.5,-1.5) node {$I_-$};
\draw (-2.2,1.9) node {$I_+$};
\end{tikzpicture}\caption{Intervals $I_\pm$.}
\label{fig:intervals}
\end{figure}
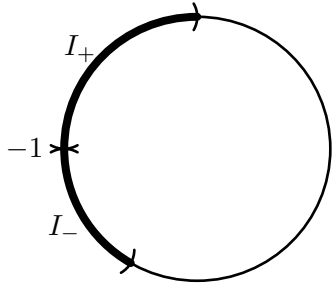

 Then it holds that
 \begin{align*}
  \Ad e^{iT(g_-)}(T(f_-)) = T(\Exp (g_-)_*(f_-)) + \beta(\Exp (g_-), f_-),
 \end{align*}
 where $\beta(\Exp (g_-), f_-)$ is defined by a similar formula as before:
 \begin{align}\label{eq:betanonsmooth}
   \beta(\Exp (g_-), f)
   \coloneqq \frac{c}{24\pi}\int_{\supp g_-}\{\Exp (g_-),z\}\bigg\vert_{z = e^{i\theta}}f(e^{i\theta})e^{i2\theta}d\theta,
 \end{align}
 where the integral is restricted to $\supp g_-$ on which the Schwarzian derivative is defined.
\end{lemma}
\begin{proof}
Let $t \in \RR$.
Since $f_-$ is piecewise smooth and $C^1$ and $g$ is smooth, by Lemma \ref{lm:smoothaction} we have
\[
   \Ad e^{iT(g)}(e^{iT(tf_-)}) = e^{iT(\Exp (g)_*(tf_-))} e^{i\beta(\Exp(g), tf_-)}.
\]
Furthermore, note that
\begin{align*}
 \beta(\Exp(g), tf_-)
 &= \frac{c}{24\pi}\int_0^{2\pi}\{\Exp (g),z\}\bigg\vert_{z = e^{i\theta}}\;tf_-(e^{i\theta})e^{i2\theta}d\theta \\
 &= \frac{c}{24\pi}\int_{\supp g_-}\{\Exp (g_-),z\}\bigg\vert_{z = e^{i\theta}}\;tf_-(e^{i\theta})e^{i2\theta}d\theta \\
 &=  \beta(\Exp(g_-), tf_-),
\end{align*}
because $\Exp(g_-)$ and $f_-$ has support contained in a common interval and
$\Exp(g_-)$ is smooth there.

Note that $g_\pm \in \tremezzi$, hence $e^{iT(tf_\pm)}$ and $e^{iT(g_\pm)}$ are affiliated to $\A(I_\pm)$
by \cite[Proposition 2.3]{Weiner06}, and it follows that $e^{iT(g)} = e^{iT(g_-)}e^{iT(g_+)}$.
By the assumed support property, we have
\[
   \Ad e^{iT(g)}(e^{iT(tf_-)})
   = \Ad (e^{iT(g_-)}\cdot e^{iT(g_+)})(e^{iT(tf_-)})
   = e^{iT(\Exp (g_-)_*(tf_-))} \cdot e^{i\beta(\Exp(g_-),tf_-)}.
\]
By taking the derivative with respect to $t$, we obtain
 \begin{align*}
  \Ad e^{iT(g)}(T(f_-)) = \Ad e^{iT(g_-)}(T(f_-)) = T(\Exp (g_-)_*(f_-)) + \beta(\Exp (g_-), f_-),
 \end{align*}
on the full domain.

\end{proof}

\begin{theorem}\label{th:positivitysol}
A soliton $\s$ of a conformal net $\mathcal{A}$ is $\psoneone(S^1)$-covariant,
the translations act continuously and have positive energy.
\end{theorem}
\begin{proof}
Our strategy is to write the translation as a product of three elements localized in
half-lines or an interval, in each of which $\sigma$ is normal. For this purpose,
let $I_{(\theta_1, \theta_2)} = \{e^{i\theta}:\theta_1 < \theta < \theta_2\}$ be an interval on $S^1$
and we take a $C^{\infty}$-function
$h_+:S^1\setminus\{-1\}\rightarrow\RR$ which is equal to $0$ on $I_{(-\pi, 0)}$
and equal to $1$ on $I_{(\frac\pi 2,\pi)}$. Similarly, let $h_-(x)$ be a $C^\infty$-function
which is equal to $1$ on $I_{(-\pi,-\frac\pi 2)}$ and equal to $0$ on $I_{(0,\pi)}$.
They have disjoint supports.

Let us first prove the following relation:
\begin{align}\label{eq:nonsmoothtransformation}
\Ad e^{itT(h_-\gt)}(T(\gt)) = T(\Exp(th_-\gt)_* (\gt))+\beta(\Exp(th_-\gt),\gt),
\end{align}
where $\gt(e^{i\theta}) = 1+\cos\theta$ is the generator of translations (see Section \ref{smooth})
Note that $h_- \gt$ is supported in a certain interval $I_-$, one of whose boundary
is $-1$, hence so is $\Exp (t h_- \gt)$.
We decompose $\gt$ into two pieces $\gt_+, \gt_- \in \tremezzi$
such that $\gt_-(\theta) = \gt(\theta)$ on $I_-$ and $\gt_+ = \gt - \gt_-$.
Note that $\beta(\Exp(t h_-\gt),\gt_-) = \beta(\Exp(t h_-\gt),\gt)$,
since the supports of $\Exp(t h_-\gt)$ and of $\gt_+$ are disjoint, see \eqref{eq:betanonsmooth}.
As $h_-\gt$ coincides with $\gt$ on an interval, one of whose boundary point is $-1$,
we can apply Lemma \ref{lm:nonsmoothaction} to obtain
\begin{align}\label{eq:hminus}
\Ad e^{itT(h_-\gt)}(T(\gt_-)) &= T(\Exp(t h_-\gt)_* (\gt_-))+\beta(\Exp(t h_-\gt),\gt_-) \nonumber\\
&= T(\Exp(t h_-\gt)_* (\gt_-))+\beta(\Exp(t h_-\gt),\gt).
\end{align}
One the other hand, since $h_-\gt$ and $\gt_+$ have disjoint support (see Figure \ref{fig:ht}), we have
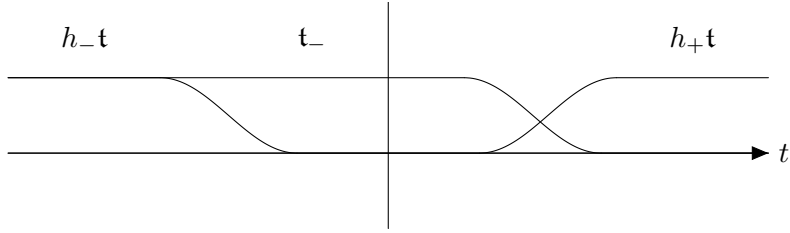
\begin{figure}[ht]
\centering
\begin{tikzpicture}[line cap=round,line join=round,>=triangle 45,x=1.0cm,y=1.0cm, scale = 1]
\clip(-6,-1) rectangle (6,2);
      \draw[->] (-5,0) -- (5,0) node[right] {$t$};
      \draw[->] (0,-1) -- (0,5);
      \draw[domain=-5:-3,smooth,variable=\x,black] plot ({\x},1);
      \draw[domain=-3:-1.2,smooth,variable=\x,black] plot ({\x}, {(cos ((\x+3)*100)+1)/2});
      \draw[domain=-1.2:5,smooth,variable=\x,black] plot ({\x},0);

      \draw[domain=-5:1,smooth,variable=\x,black] plot ({\x},1);
      \draw[domain=1:2.8,smooth,variable=\x,black] plot ({\x}, {(cos ((\x-1)*100)+1)/2});
      \draw[domain=2.8:5,smooth,variable=\x,black] plot ({\x},0);

      \draw[domain=-5:1.2,smooth,variable=\x,black] plot ({\x},0);
      \draw[domain=1.2:3,smooth,variable=\x,black] plot ({\x}, {(cos ((\x-3)*100)+1)/2});
      \draw[domain=3:5,smooth,variable=\x,black] plot ({\x},1);
      \node at(-4,1.5) {$h_-\gt$};
      \node at(-1,1.5) {$\gt_-$};
      \node at(4,1.5) {$h_+\gt$};
\end{tikzpicture}\caption{Vector fields $h_-\gt, \gt_-, h_+\gt$ in the real line picture.}
\label{fig:ht}
\end{figure}

\begin{align}\label{eq:hplus}
\Ad e^{itT(h_-\gt)}(T(\gt_+)) &= T(\gt_+).
\end{align}
Note that $\Exp(t h_-\gt)_*\gt = \Exp(t h_-\gt)_*\gt_+ + \Exp(t h_-\gt)_*\gt_-
= \gt_+ + \Exp(t h_-\gt)_*\gt_-$. By adding the sides of \eqref{eq:hminus} and \eqref{eq:hplus}, we obtain
on the intersection of the domains
\begin{align*}
\Ad e^{itT(h_-\gt)}(T(\gt)) &= T(\Exp(t h_-\gt)_* (\gt))+\beta(\Exp(t h_-\gt),\gt).
\end{align*}
The intersection of the operators in \eqref{eq:hminus}\eqref{eq:hplus} includes $C^\infty(L_0) = \bigcap_n \dom(L_0^n)$,
hence so does the sum. The right-hand side of the last expression is essentially self-adjoint
by \cite[Theorem 4.4]{CW05} (cited in Section \ref{piecewise}).
Hence the left-hand side is a self-adjoint extension of the right-hand side,
and therefore, they must coincide on the full domain.

Next, we write $e^{itT({\gt})}$ as
\[
 e^{itT({\gt})} = e^{itT(h_-\gt)}\cdot e^{-itT(h_-\gt)}e^{itT(\gt)}e^{-itT(h_+\gt)}\cdot e^{itT(h_+\gt)}.
\]
We claim that $e^{-itT(h_-\gt)}e^{itT(\gt)}e^{-itT(h_+\gt)}$ is localized on an interval
whose closure does not contain $-1$ (such an interval depends on $t$). 
This follows from \eqref{eq:nonsmoothtransformation}.
Indeed, $\Exp(th_-\gt)_* (\gt)$ agrees with $\gt$ in a neighborhood of $-1$ (depending on $t$) and
\begin{align*}
e^{-itT(h_-\gt)}e^{itT(\gt)}e^{-itT(h_+\gt)}
&=e^{itT(\Exp(h_-\gt)_* (\gt))}e^{i\beta(\Exp(th_-\gt),\gt)}\cdot e^{-itT(h_-\gt)}e^{-itT(h_+\gt)}\\
&= e^{itT(\Exp(h_-\gt)_* (\gt))}e^{i\beta(\Exp(th_-\gt),\gt)}e^{-itT(h_-\gt + h_+\gt)}, 
\end{align*}
where we used the linearity of $T$ on functions of class $\tremezzi$,
and the last expression is localized in a bounded interval:
as $h_-\gt + h_+\gt$ equals $\gt$ in a neighborhood of $-1 \in S^1$,
$\Ad e^{-itT(h_-\gt + h_+\gt)}$ implements the same action on
$\A(I_{t,\epsilon})$ for some neighborhood $I_{t,\epsilon}$
for small $t$ as the action of $\Ad e^{itT(\Exp(h_-\gt)_* (\gt))}$.
In other words, $\Ad e^{itT(\Exp(h_-\gt)_* (\gt))}e^{-itT(h_-\gt + h_+\gt)}$
is trivial on $\A(I_{t,\epsilon})$, which implies that
$e^{itT(\Exp(h_-\gt)_* (\gt))}e^{-itT(h_-\gt + h_+\gt)}$ is localized in
$I_{t,\epsilon}'$.

We introduce a representation of the translation group by
\begin{align*}
U_\sigma (t):=\sigma(e^{itT(h_-\gt)})\sigma(e^{-itT(h_-\gt)}e^{itT(\gt)}e^{-itT(h_+\gt)})\sigma(e^{itT(h_+\gt)})
\end{align*}
Note first the relation
$\beta(\gamma_1\circ\gamma_2,f)=\beta(\gamma_1,\gamma_{2 *}(f))+\beta(\gamma_2,f)$ for smooth $\gamma$,
which is equivalent to $0=\beta(\id, f) = \beta(\gamma^{-1},\gamma_*f)+\beta(\gamma,f)$.
This continues to hold for nonsmooth $\gamma = \Exp(th_\pm \gt)$ and $f= \gt$.
Indeed, $\Exp(th_\pm \gt)$ can be continued to a smooth diffeomorphisms with compact support, say $I$,
and let $I_\pm$ be the intervals obtained by removing $-1$ from $I$.
The function $\beta(\Exp(th_\pm \gt), \gt)$ is defined by the integral \eqref{eq:betanonsmooth},
which can be extended to $\gt_\pm$, the restriction of $\gt$ to $I_\pm$.
Hence the relation $\beta(\Exp(-t h_\pm\gt),\Exp(t_2 h_\pm\gt)_{*}(\gt_\pm))+\beta(\Exp(t_2 h_-\gt),\gt_\pm) = 0$
holds. Again by linearity in the second variable, we have for $t_1,t_2\in\RR$
\begin{align}\label{eq:betavanishing}
\begin{array}{rl}
0&=\beta(\id,\gt)=\beta(\Exp(-t_2 h_-\gt),\Exp(t_2 h_-\gt)_{*}(\gt))+\beta(\Exp(t_2 h_-\gt),\gt),\\
0&=\beta(\id,\gt)=\beta(\Exp(-t_1 h_+\gt),\Exp(t_1 h_+\gt)_{*}(\gt))+\beta(\Exp(t_1 h_+\gt),\gt).
\end{array}
\end{align}
By recalling that $h_-$ and $h_+$ have disjoint supports, with the help of \eqref{eq:nonsmoothtransformation}
and an analogous relation for $h_+$,
this yields a one-parameter group in $t$:
\begin{align*}
& U_\sigma (t_1)U_\sigma (t_2)\\
&=\sigma(e^{it_1T(h_-\gt)})\sigma(e^{-it_1T(h_-\gt)}e^{it_1T(\gt)}e^{-it_1T(h_+\gt)})\sigma(e^{it_1T(h_+\gt)})\\
&\quad\cdot \sigma(e^{it_2T(h_-\gt)})\sigma(e^{-it_2T(h_-\gt)}e^{-it_2T(\gt)}e^{it_2T(h_+\gt)})\sigma(e^{it_2T(h_+\gt)}) \\
&=\sigma(e^{it_1T(h_-\gt)})\sigma(e^{it_2T(h_-\gt)})\sigma(e^{-it_1T(h_-\gt)}e^{it_1T(\Exp(t_2 h_-\gt)_*(\gt))}e^{i\beta(\Exp(t_2 h_-\gt),\gt)}e^{-it_1T(h_+\gt)}) \\
&\quad\cdot \sigma(e^{it_1T(h_+\gt)}) \sigma(e^{-it_2T(h_-\gt)}e^{it_2T(\gt)}e^{-it_2T(h_+\gt)})\sigma(e^{it_2T(h_+\gt)}) \\
&=\sigma(e^{i(t_1+t_2)T(h_-\gt)})\sigma(e^{-it_1T(h_-\gt)}e^{it_1T(\Exp(t_2 h_-\gt)_*(\gt))}e^{i\beta(\Exp(t_2 h_-\gt),\gt)}e^{-it_1T(h_+\gt)})\\
&\quad\cdot \sigma(e^{-it_2T(h_-\gt)}e^{it_2T(\Exp(t_1h_+)_*\gt)}e^{i\beta(\Exp(t_1 h_+\gt),\gt)}e^{-it_2T(h_+\gt)})\sigma(e^{it_1T(h_+\gt)})\sigma(e^{it_2T(h_+\gt)}) \\
&=\sigma(e^{i(t_1+t_2)T(h_-\gt)})\\
&\quad\cdot \sigma(e^{-it_1T(h_-\gt)}e^{it_1T(\Exp(t_2 h_-\gt)_*(\gt))}e^{-it_1T(h_+\gt)}e^{-it_2T(h_-\gt)}e^{it_2T(\Exp(t_1h_+)_*\gt)}e^{-it_2T(h_+\gt)})\\
&\quad\cdot \sigma(e^{it_1T(h_+\gt)})\sigma(e^{it_2T(h_+\gt)})
\cdot e^{i\beta(\Exp(t_2 h_-\gt),\gt)}e^{i\beta(\Exp(t_1 h_+\gt),\gt)}\\
&=\sigma(e^{i(t_1+t_2)T(h_-\gt)})\\
&\quad\cdot \sigma(e^{-i(t_1+t_2)T(h_-\gt)}e^{it_1T(\gt)}e^{i\beta(\Exp(-t_2 h_-\gt),\Exp(t_2 h_-\gt)_{*}(\gt))}
e^{-it_1T(h_+\gt)}e^{it_2T(\Exp(t_1h_+)_*\gt)}e^{-it_2T(h_+\gt)})\\
&\quad\cdot \sigma(e^{it_1T(h_+\gt)})\sigma(e^{it_2T(h_+\gt)})
\cdot e^{i\beta(\Exp(t_2 h_-\gt),\gt)}e^{i\beta(\Exp(t_1 h_+\gt),\gt)}\\
&=\sigma(e^{i(t_1+t_2)T(h_-\gt)})\\
&\quad\cdot \sigma(e^{-i(t_1+t_2)T(h_-\gt)}e^{it_1T(\gt)}
e^{it_2T(\gt)}e^{i\beta(\Exp(-t_1 h_+\gt),\Exp(t_1 h_+\gt)_{*}(\gt))}e^{-i(t_1+t_2)T(h_+\gt)})\\
&\quad\cdot \sigma(e^{it_1T(h_+\gt)})\sigma(e^{it_2T(h_+\gt)})
\cdot e^{i\beta(\Exp(t_2 h_-\gt),\gt)}e^{i\beta(\Exp(-t_2 h_-\gt),\Exp(t_2 h_-\gt)_{*}(\gt))}e^{i\beta(\Exp(t_1 h_+\gt),\gt)}\\
&=\sigma(e^{i(t_1+t_2)T(h_-\gt)})\sigma(e^{-(t_1+t_2)T(h_-\gt)}e^{i(t_1+t_2)T(\gt)}e^{-i(t_1+t_2)T(h_+\gt)})\\
&\quad\cdot \sigma(e^{i(t_1+t_2)T(h_+\gt)})\cdot e^{i\beta(\Exp(t_2 h_-\gt),\gt)}e^{i\beta(\Exp(-t_2 h_-\gt),\Exp(t_2 h_-\gt)_{*}(\gt))}\\
&\quad\cdot e^{i\beta(\Exp(t_1 h_+\gt),\gt)}e^{i\beta(\Exp(-t_1 h_+\gt),\Exp(t_1 h_+\gt)_{*}(\gt))} \\
&=U_\sigma (t_1+t_2),
\end{align*}
where the scalars cancel by \eqref{eq:betavanishing}.
To show continuity of $U^\sigma(t)$, it is enough to see that $U^\sigma(t) \to \1$
as $t\to 0$ in the strong operator topology.
All the factors in the product
\[
 U^\sigma(t) = \sigma(e^{itT(h_-\gt)})\sigma(e^{-itT(h_-\gt)}e^{itT(\gt)}e^{-itT(h_+\gt)})\sigma(e^{itT(h_+\gt)})
\]
are unitary, and hence uniformly bounded.
The first and the third factors tend to $\1$ by normality of $\sigma$,
as they are supported in a fixed half-line.
As for the second factor, the product $e^{-itT(h_-\gt)}e^{itT(\gt)}e^{-itT(h_+\gt)}$
tends to $\1$ in the vacuum representation, and this is localized in a fixed interval
for small $t$, hence by normality of $\sigma$ the second factor tends to $\1$ as well.

It remains to prove the positivity of energy.
We do this by showing that $U_\sigma (t)$ can be obtained as a limit
in the strong resolvent sense of a sequence of one-parameter unitary groups
with positive generator.
Let $\gt_1$ be a $C^{\infty}$-vector field on $S^1$ such that $C_*(\gt_1)$ is equal to $1$ on
$(-\infty,1)$ and equal to $0$ on $(2,+\infty)$. From $\gt_1$ we construct a sequence of vector fields
(see Figure \ref{fig:tn})
\begin{figure}[ht]
\centering
\begin{tikzpicture}[line cap=round,line join=round,>=triangle 45,x=1.0cm,y=1.0cm, scale = 1]
\clip(-7,-1) rectangle (7,2);
      \draw[->] (-7,0) -- (7,0) node[below left] {$t$};
      \draw[->] (0,-1) -- (0,5);
      \draw[domain=-7:1,smooth,variable=\x,black] plot ({\x},1);
      \draw[domain=1:2.8,smooth,variable=\x,black] plot ({\x}, {(cos ((\x-1)*100)+1)/2});
      \draw[domain=2.8:5,smooth,variable=\x,black] plot ({\x},0);

      \draw[domain=-7:2,smooth,variable=\x,black] plot ({\x},1);
      \draw[domain=2:5.6,smooth,variable=\x,black] plot ({\x}, {(cos ((\x/2-1)*100)+1)/2});
      \draw[domain=5.6:6,smooth,variable=\x,black] plot ({\x},0);
      \node at(1,1.5) {$\gt_1$};
      \node at(2,1.5) {$\gt_2$};
\end{tikzpicture}\caption{Vector fields $\gt_n$ in the real line picture.}
\label{fig:tn}
\end{figure}
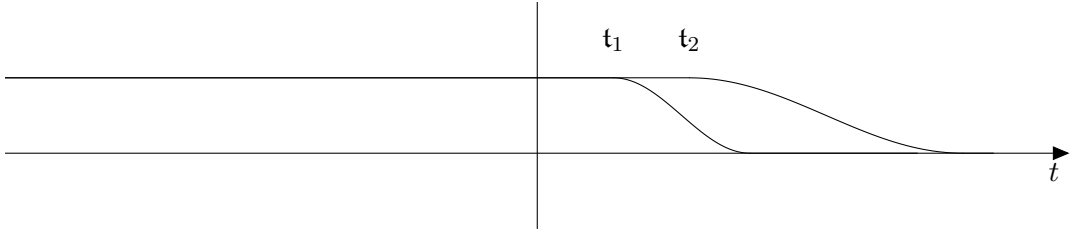
\begin{align*}
C_*(\gt_n)(t)\coloneqq C_*(\gt_1)\left(\frac{t}{n}\right),\qquad t\in\RR,\hbox{ }n\in\NN.
\end{align*}
We fix $2 < s < \frac52$ (cf.\! the proof of Lemma \ref{lm:smoothaction}).
Let us show that $\gt_n\rightarrow\gt$ in the $H^s(S^1)$-topology.
For this it is sufficient to show that $\left\lbrace\frac{d^3}{d\theta^3}\gt_n\right\rbrace_{n\in\NN}$ is a sequence of functions in $L^1(S^1)$ uniformly bounded in $n$ and that $\gt_n\rightarrow \t$ in $L^1(S^1)$:
this implies that $|k^3 \hat \gt_n(k)| < \mathrm{Const.}$, where $\hat \gt_n(k)$ is the $k$-th Fourier
coefficient of $\gt_n$, or equivalently, $|k^{2s} \hat \gt_n(k)^2| < \frac{\mathrm{Const.}}{k^{6-2s}}$,
and the right-hand side is summable in $k$ since $6-2s > 1$.
From the convergence $\gt_n \to \gt$ in $L^1$ we obtain the convergence of each $\hat \gt_n(k)$,
Therefore, by the Lebesgue dominated convergence theorem
(applied to the measurable set $\ZZ$ with the counting measure),
we obtain the convergence $\gt_n \to \gt$ in $H^s(S^1)$.
Let us make these necessary estimates separately.

Explicitly, as $\frac{dC}{d\theta}(e^{i\theta}) = \frac1{1+\cos\theta}$, we have
\begin{align*}
\gt_n(e^{i\theta}) = (1+\cos\theta)\cdot C_*(\gt_1)\left(\frac{C(e^{i\theta})}{n}\right)
\end{align*}
and $C_*(\gt_1)$ is bounded.
The third derivative of $\gt_n$ is 
\begin{align}\label{3derivative}
\frac{d^3}{d\theta^3}\gt_n(e^{i\theta}) &= \sin \theta\cdot C_*(\gt_1)\left(\frac{C(e^{i\theta})}{n}\right)
-\frac{2\cos\theta}{n(1+\cos\theta)}\frac{d}{dt}(C_*(\gt_1))\left(\frac{C(e^{i\theta})}{n}\right) \nonumber \\
 &\quad-\frac{\sin^2 \theta}{n(1+\cos\theta)^2}\frac{d}{dt}(C_*(\gt_1))\left(\frac{C(e^{i\theta})}{n}\right)
 +\frac{1}{n^3(1+\cos\theta)^2}\frac{d^3}{dt^3}(C_*(\gt_1))\left(\frac{C(e^{i\theta})}{n}\right).
\end{align}
Recall that $C_*(\gt_1)$ is constant on $(-\infty, 1) \cup (2,\infty)$.
The first term of the right-hand side of \eqref{3derivative} is clearly uniformly bounded in $n$ on $S^1$. For the second term of the right-hand side of \eqref{3derivative}, by the change of variable $t = C(e^{i\theta})$, we have:
\begin{align*}
\int_0^{2\pi}\left\vert\frac{2\cos \theta }{n(1+\cos\theta)}\frac{d}{dt}(C_*(\gt_1))\left(\frac{C(e^{i\theta})}{n}\right)\right\vert d\theta 
&= \frac1n\int_n^{2n}\left|2\cos(-i\log C^{-1}(t))\frac{d}{dt}(C_*(\gt_1))\left(\frac{t}{n}\right)\right\vert dt
\end{align*}
which is uniformly bounded in $n$. Next, by calculating directly one has
$1+\cos (-i\log C^{-1}(t)) = \frac2{1+t^2}$ and $\sin(-i\log C^{-1}(t)) = \frac{2t}{1+t^2}$, hence the third term is
\begin{align*}
&\int_0^{2\pi}\left|\frac{2\sin^2\theta}{n(1+\cos \theta)^2}\frac{d}{dt}(C_*(\gt_1))\left(\frac{C(e^{i\theta})}{n}\right)\right| d\theta
=\int_n^{2n}\left|\frac{2\sin^2\theta}{n(1+\cos\theta)}\frac{d}{dt}(C_*(\gt_1))\left(\frac{t}{n}\right)\right| dt\\
&=\frac1n\int_n^{2n}\left|\frac{4t^2}{1+t^2}\frac{d}{dt}(C_*(\gt_1))\left(\frac{t}{n}\right)\right| dt
\end{align*}
which is uniformly bounded in $n$, since the integrand is bounded.
The fourth term is also uniformly bounded in $n$ since
\begin{align*}
&\int_0^{2\pi}\left|\frac{1}{n^3(1+\cos\theta)^2}\frac{d^3}{dt^3}(C_*(\gt_1))\left(\frac{C(e^{i\theta})}{n}\right)\right|d\theta 
=\int_n^{2n}\left|\frac{1+t^2}{2n^3}\frac{d^3}{dt^3}(C_*(\gt_1))\left(\frac{t}{n}\right)\right| dt \\
&\le \frac1n\int_n^{2n}\left|\frac{1+4n^2}{2n^2}\frac{d^3}{dt^3}(C_*(\gt_1))\left(\frac{t}{n}\right)\right| dt
\end{align*}
Finally, we show that $\gt_n\rightarrow\gt$ in $L^1(S^1)$ using the boundedness of $C_*(\gt_1)$:
\begin{align*}
&\int_0^{2\pi}\left| \gt(e^{i\theta})-\gt_n(e^{i\theta})\right| d\theta
=\int_0^{2\pi}\left| (1+\cos\theta)\left(1-C_*(\gt_1)\left(\frac{C(e^{i\theta})}{n}\right)\right)\right| d\theta\\
&=\int_n^{+\infty}\left| (1+\cos(-i\log C^{-1}(t)))^2\left(1-C_*(\gt_1)\left(\frac{t}{n}\right)\right)\right| dt \\
&=\int_n^{+\infty}\left| \frac4{(1+t^2)^2}\left(1-C_*(\gt_1)\left(\frac{t}{n}\right)\right)\right| dt
\longrightarrow 0 \qquad(\text{as } n\rightarrow\infty)
\end{align*}
This completes the proof that $\gt_n\to\gt$ in $H^s(S^1), 2 < s < \frac52$.

For each fixed $t$,
the representation $U_\sigma (t)$ can be obtained as the limit of $\sigma(e^{itT(\gt_n)})$
in the strong topology. Indeed,
\[
 \sigma(e^{itT(\gt_n)}) = \sigma(e^{itT(h_-\gt_n)})\sigma(e^{-itT(h_-\gt_n)}e^{itT(\gt_n)}e^{-itT(h_+\gt_n)})\sigma(e^{itT(h_+\gt_n)})
\]
Note that $h_-, h_+, \gt_n$ belong to $H^s(S^1)$, and the product is (jointly) continuous
\cite[Lemma 2.4]{CDIT18+}\cite[Lemma B.4]{IKT13} as $\frac32 < 2 < s$, hence
both $h_-\gt_n$ and $h_+\gt_n$ are convergent in $H^s(S^1)$, and by the argument of
Lemma \ref{lm:smoothaction}, they are convergent in $\tremezzi$,
hence the corresponding operators are convergent in the strong resolvent sense.
Furthermore, each of these sequences are localized in a fixed interval or a half line,
hence by the normality of $\s$ on half lines, the convergence follows.
In other words, if we write $U^\sigma(t) =: e^{itT^\sigma}$,
then $\sigma(T(\gt_n))$ (the generator of $\sigma(e^{itT(\gt_n)})$, which is defined by local normality)
is convergent to $T^\sigma$ in the strong resolvent sense.

We have by Theorem \ref{th:qei} that $T(\gt_1)\geq \alpha$ for some $\a\in\RR$.
By the fact that the Schwarzian derivative of a M\"obius transformation is 0,
it follows that the quantum energy inequalities of Theorem \ref{th:qei}
are invariant under dilations, and therefore,
\[
 T(\delta^n_*(\gt_1))=T(n\gt_n)\geq \alpha,
\]
which implies
\[
 T(\gt_n)\geq \frac{\alpha}{n}.
\]
Since $T(\gt_n)$ is localized on a half-line, by local normality of $\sigma$,
we have $\sigma(T(\gt_n)) \ge \frac\a n$.
By \cite[Theorem VIII.23]{RSII} and the convergence in the strong resolvent sense,
$T^\sigma$ is positive as well. 

By Proposition \ref{pr:psone},
the net $(\A, U, \Omega)$ is locally $\psone(S^1)$-covariant.
Any element $\g \in \psoneone(S^1)$
can be decomposed into a product $\g = \g_-\circ(\g_-^{-1}\circ\g\circ\g_+^{-1})\circ\g_+$, where
$\g_\pm \in\psoneone(S^1)$
as in the proof for $U^\sigma(t)$. It is straightforward to see that
the definition
\[
 U^\sigma(\g) = \sigma(U(\g_-))\sigma(U(\g_-^{-1}\circ\g\circ\g_+^{-1}))\sigma(U(\g_+))
\]
does not depend on the decomposition of $\g$.
If $I$ is a left half-line, we can choose $\g_-$ such that $I \cap \supp \g_+ = \emptyset$.
Now for $x \in \A(I)$ the covariance $\sigma(\Ad U(\g)(x)) = \Ad U^\sigma(\g)(\sigma(x))$
follows because the both sides are localized in $I_-$, and by the definition
$U^\sigma(\g\circ \g_+^{-1}) = \sigma(U(\g\circ\g_+^{-1}))$.
\end{proof}

\subsection{Solitons from nonsmooth diffeomorphisms}\label{nonsmoothsoliton}
Here we construct families of continuously many proper solitons for any conformal net $\A$, using the diffeomorphism covariance.
\paragraph{Nonsmooth extendable diffeomorphisms.}
Let $\diff(S^1,-1)\subset \diff^0(S^1) $ be the class of orientation preserving homeomorphism $\nu$ of $S^1$,
which have the following properties
\begin{itemize}
\item $\nu(-1)=-1$,
\item $\nu$ is a smooth map from $S^1\setminus \{-1\}$ onto $S^1\setminus \{-1\}$,
and the left and right derivatives at all orders exist at the point $-1$, and the first order left and right derivatives are nonzero.
\end{itemize}
By Borel's theorem \cite[Theorem 1.2.6]{Hoermander90},
for any sequence of real numbers $\{\lambda_n\}_{n \ge 1}$,
there is a smooth function $f$ such that $f(-1) = -1$ and
$\frac{d^n f}{d\theta^n}(-1) = \lambda_n$.
Let us take $\lambda_n$ as the left derivatives of $\nu$ at $\theta = -1$.
Then we can find a smooth function $f$ such that $f(-1) = -1$
$\frac{d^n f}{d\theta^n}(-1) = \lambda_n$.
Let $I$ be an interval containing $-1$ and call $I_-, I_+$ the subintervals
resulting from $I$ by removing $-1$.
There is a function $\tilde f$ which agrees with $\nu$ on $I_-$
and with $f$ on $I_+$.
This is smooth even at $\theta = -1$, because all the left and right derivatives coincide.
Now, since the first derivative at $\theta = -1$ is nonzero, it defines a diffeomorphism
of a small interval containing $\overline{I_-}$, hence the restriction to this small interval
can be continued to a diffeomorphism of $S^1$ (see \cite[Lemma 3.9]{CDIT18+}).
Let us call it $\nu_{I_-}$, then $\nu_{I_-} \in \diff(S^1)$ and
agrees with $\nu$ on $I_-$.
Similarly, we can find $\nu_{I_+} \in \diff(S^1)$ which agrees with $\nu$ on $I_+$.

\paragraph{Irreducible solitons.}
Let $\mathcal{A}$ be a conformal net on $S^1$ on the Hilbert space $\mathcal{H}$ and $U$ its associated projective representation of
$\diff(S^1)$.
For $\nu\in\diff(S^1,-1)$ and for every $I\in{\overline{\I}_\RR}$ we choose $\nu_I\in\diff(S^1)$ which agrees with $\nu$ on $I$
(there is such $\nu_I$ even if one of the endpoint of $I$ is $-1$ (half-lines in the $\RR$ picture)
by the remark above).
We denote by $\sigma_\nu$ the family of maps $\sigma_\nu\coloneqq\{\sigma_\nu^I\}$ where 
\begin{align*}
\sigma_\nu^I:\mathcal{A}(I)&\longrightarrow \B(\H)\\
 x&\longmapsto\sigma_\nu^I(x)\coloneqq\Ad U(\nu_I)(x)
\end{align*}
and $I\in{\overline{\I}_\RR}$, $\nu\in\diff(S^1,-1)$. 

\begin{proposition}\label{pr:welldefsol}
Let $\nu\in\diff(S^1,-1)$. Then $\sigma_\nu=\{\sigma_\nu^I\}$ is an irreducible soliton of
the conformal net $\mathcal{A}$ with index $1$.
\end{proposition}
\begin{proof}
Normality on each $I$ follows because each map $\sigma_\nu^I$ is given
by the adjoint action $\Ad U(\nu_I)$.
We show that the family of maps $\sigma_\nu$ is compatible, namely that, if $I\subset \hat I$ for $I,\hat I\in{\overline{\I}_\RR}$,
then $\sigma_\nu^{\hat I}\restriction_{\mathcal{A}(I)}=\sigma_\nu^I$.
By definition, $\nu_I, \nu_{\hat I} \in\diff(S^1)$ agree with $\nu$ on $I$ and $\hat I$,
respectively, hence they agree on $I$.
Then on $\mathcal{A}(I)$ we have 
\[
\Ad U(\nu_I)= \Ad U(\nu_{\hat I})\circ \Ad U(\nu_{\hat I}^{-1}\circ\nu_I) = \Ad U(\nu_{\hat I}),
\]
because $\nu_{\hat I}^{-1}\circ\nu_I$ is a diffeomorphism of the circle localized in
$I'$ and in this case $\Ad U(\nu_{\hat I}^{-1}\circ\nu_I)$ acts trivially on $I$ by \eqref{eq:diffcov2}.

Irreducibility follows because
$\bigvee_{I\in {\overline{\I}_\RR}}\sigma(\A(I)) = \bigvee_{I\in {\overline{\I}_\RR}}\sigma(\A(\nu I)) = \bigvee_{I\in {\overline{\I}_\RR}}\A(I)$
and the weak closure of the right-hand side is $\B(\H)$ by (CN\ref{cn:irreducibility}).
The index of $\sigma_\nu$ is $1$ because
\[
 \sigma_\nu(\A(\RR_+)) = \A(\nu \RR_+) = \A((\nu \RR_+)')' = \A(\nu \RR_-)' = \sigma_\nu(\A(\RR_-))'.
\]
\end{proof}

Now we show that if $\nu$ has different left and right derivatives,
then $\sigma_\nu$ is a proper soliton.
Modular theory is used as a tool to show non-triviality of the constructed soliton.
Let us introduce the notation for left and right derivatives (they are in $\RR_+$ because
$\nu$ is orientation-preserving):
\[
 \partial_\pm \nu(-1) = -i\lim_{\theta \to 0^\pm}\frac{\nu(-e^{i\theta})-\nu(-1)}{\theta}, 
\]
We also denote their ratio by $r(\nu) := \partial_+ \nu(-1)/\partial_- \nu(-1) \in \RR_+$.
For this purpose, we need the results from Appendix \ref{c1pws}.
Actually, in order to show that there are proper solitons,
we can take $\nu$ which coincides with dilations with different parameter
in a neighborhood of $-1$. In this way, $\nu_\pi$ is a product of $\psi_t$ below
for some $t$ and an element in $\diff(S^1)$, and one does not have to invoke Proposition \ref{pr:psone}.

\begin{theorem}\label{th:typeIsol}
Let $\nu\in\diff(S^1,-1)$ and assume that $r(\nu) \neq 1$. Then $\sigma_\nu$ is a proper, irreducible,
$\psonezero(S^1)$-covariant soliton of $\A$.
Furthermore, let $\nu_1,\nu_2 \in\diff(S^1,-1)$. Then
$\sigma_{\nu_1}$ and $\sigma_{\nu_2}$ are unitarily equivalent if and only if
$r(\nu_1) = r(\nu_2)$.
\end{theorem}
\begin{proof}
Let us first consider the function 
\begin{align*}
\psi_t(e^{i\theta})\coloneqq\begin{cases} 
      e^{i\theta} & \text{if } \theta\in [-\pi,0) \\
      \delta(t)(e^{i\theta}) & \text{if } \theta\in[0,\pi) 
   \end{cases}
\end{align*}
which is smooth except $-1,1$. It is constructed just by gluing the identity map on
$[-\pi,0)$ and the dilation on $[0,\pi)$ in the circle picture (they correspond to $\RR_+$ and $\RR_-$
in the $\RR$ picture, respectively).
By passing to the $\RR$ picture (which does not affect $r(\nu)$),
it is immediate that $r(\psi_t) = e^{-t}$.

We may assume that $\nu(1)=1$, because if $\nu(1)\neq 1$, then
there is a smooth diffeomorphism $\underline \nu$ such that $\supp \underline\nu$ does not contain $-1$ and $\underline \nu(1) = \nu(1)$.
We then have $\underline \nu^{-1}\circ \nu(1) = 1$.
As such $\underline \nu^{-1}$ is represented by $U(\underline \nu^{-1})$,
the questions of properness, irreducibility and covariance of $\sigma_\nu$
are equivalent to that of $\sigma_{\underline \nu^{-1}\circ\nu} = \Ad U(\underline \nu^{-1})\circ \s_\nu$.

From $\nu\in\diff(S^1,-1)$ with $\nu(1)=1$, we construct a homeomorphism $\nu_\pi$ of $S^1$
which is smooth except two points, the points $-1$ and $1$:
\begin{align*}
\nu_\pi \coloneqq \nu\circ R_{\pi}\circ \nu^{-1}\circ R_{\pi},
\end{align*}
where $R_{\pi}$ the rotation by $\pi$.
If $t = -\log r(\nu)$, then $\psi_t\circ \nu_\pi^{-1}\in\psone(S^1)$ because
the discrepancy of the derivatives at $1$ and $-1$ cancel exactly.

We show that $\sigma_\nu$ is a proper soliton, i.e.\! it does not extend to a  DHR representation.
Let us suppose the contrary, namely that it were a restriction of a DHR representation.
We denote the extension by $\sigma_\nu$.
Then $\sigma_\nu$ is rotation covariant \cite[Theorem 6]{DFK04}, namely there is a unitary representation of the universal covering of $S^1$, $\theta\rightarrow U^\nu (R_\theta)$, such that	
\begin{align*}
\Ad U^\nu (R_\theta)\circ \sigma_\nu =\sigma_\nu\circ \Ad U(R_\theta).
\end{align*}
Furthermore, $\s_\nu$ is invertible, because $\s_\nu$ has index $1$.

Consider the following composition as a DHR representation:
\begin{align*}
\rho := \Ad U(\psi_t\circ \nu_\pi^{-1}) \circ \sigma_\nu \circ \Ad U(R_{\pi}) \circ \sigma_{\nu^{-1}}\circ \Ad U(R_{\pi})
\end{align*}
It follows that this is implemented by a unitary
$U_\rho := U(\psi_t\circ \nu_\pi^{-1})U^\nu(R_{\pi})U(R_{\pi})$ since
\begin{align*}
\rho &=\Ad U(\psi_t\circ \nu_\pi^{-1}) \circ \sigma_\nu \circ \Ad U(R_{\pi}) \circ \sigma_{\nu^{-1}}\circ \Ad U(R_{\pi}) \\
&=\Ad U(\psi_t\circ \nu_\pi^{-1}) \circ \Ad U^\nu(R_{\pi})\circ \sigma_\nu \circ \sigma_{\nu^{-1}}\circ \Ad U(R_{\pi}) \\
&=\Ad U(\psi_t\circ \nu_\pi^{-1}) \circ \Ad U^\nu(R_{\pi})\circ \Ad U(R_{\pi}).
\end{align*}

On the other hand, by construction we have that $\rho(x)=x$ for $x\in\mathcal{A}(I_{(-\pi,0)})$
where $I_{(-\pi,0)} = I_{(0,\pi)}'$
and $\rho(x)=\Ad U(\delta(t))$ for $x\in\mathcal{A}(I_{(0,\pi)})$.
Therefore, the unitary $U_\rho$ must belong to $\mathcal{A}(I_{(0,\pi)})$ by Haag duality.
At the same time, by the Bisognano-Wichmann property (CN\ref{bisognano}),
the dilation $\Ad U(\delta(t))$ is the modular automorphism $\sigma_{\A(I_{(0,\pi)})}^{t/2\pi}$
of $\mathcal{A}(I_{(0,\pi)})$ with respect to the vacuum vector $\Omega$.
This is a contradiction because the modular automorphisms for $t\neq 0$ cannot be inner for
$\mathcal{A}(I_{(0,\pi)})$ which is a type III factor
\cite[Theorem 10.29]{SZ79}.
Therefore, we conclude that $\sigma_\nu$ does not extend to a DHR representation.

Next, let $\nu_1, \nu_2 \in \diff(S^1,-1)$.
It holds that $r(\nu_1) = r(\nu_2)$, if and only if $r(\nu_1^{-1}\circ\nu_2) = 1$,
and by Proposition \ref{pr:psone} and the argument of the previous paragraphs,
if and only if $\nu_1^{-1}\circ\nu_2$ is implemented by a unitary, or in other words,
if and only if $\sigma_{\nu_1}$ and $\sigma_{\nu_2}$ are unitarily equivalent.

Finally, let us prove $\psonezero(S^1)$-covariance.
We already know from Theorem \ref{th:positivitysol} that any soliton is $\psoneone(S^1)$-covariant,
hence it is enough to show that $\s_\nu$ is dilation covariant.
This follows because $r(\d(t)\circ\nu\circ\d(-t)) = r(\nu)$,
hence $\s_\nu$ and $\s_{\d(t)\circ\nu\circ\d(-t)}$ are unitarily equivalent.
This unitary implements the dilation $\d(t)$.
\end{proof}

\paragraph{Type III solitons.}
Instead of considering functions $\nu\in\diff(S^1, -1)$, we can do a similar construction using a function
$\nu$ with the following properties:
\begin{itemize}
\item $\nu$ is smooth on $S^1\setminus \{-1\}$ and the left and right derivatives at all orders
exist at the point $-1$ and the first order left and right derivatives are nonzero.
\item $\nu$ is injective and orientation preserving.
\item $\nu(S^1\setminus \{-1\})$ is a proper interval $I_\nu$ of $S^1$.
\end{itemize} 
 
If we take such a $\nu$, $\sigma_\nu$ still yields a soliton of the conformal net $\mathcal{A}$,
since the arguments of Proposition \ref{pr:welldefsol} remain valid except for irreducibility and index.
This type of construction was implicitly presented in \cite{LX04} and \cite{KLX05},
namely, by taking the construction of soliton of the net $\A\otimes\A$ in \cite{LX04}
and restrict it to $\A\otimes \CC\1$.
In this case, one obtains solitons $\sigma_\nu$ which are of type III
(namely $\left(\bigcup_{I\subset {\overline{\I}_\RR}}\sigma_\nu(\mathcal{A}(I))\right)''$
is a type III factor).
For completeness we show that this type of construction also yields a proper soliton.
We further prove that such a soliton is locally $\widetilde{\D^s(S^1)}$-covariant with $s>3$,
namely, there is a unitary representation $U^\sigma$ of $\widetilde{\D^s(S^1)}$ such that
$\Ad U^\sigma(\g)(\A(I)) = \A(\g I)$ as long as $\g$ is contained in a neighborhood $\U$
of the unit element of $\widetilde{\D^s(S^1)}$ such that $\g I \subset \RR$ for $\g \in \U$.

\begin{theorem}\label{th:type3}
For $\nu$ as above, $\sigma_\nu$ is a proper soliton of type $\text{III}$,
locally $\widetilde{\D^s(S^1)}$-covariant with $s>3$.
For any pair $\nu_1, \nu_2$, two solitons $\sigma_{\nu_1}$ and $\sigma_{\nu_2}$ are unitarily equivalent.
\end{theorem}
\begin{proof}
We must show that the representation $\sigma_\nu$ does not extend to a representation of the net $\mathcal{A}$ of the circle. Let $I_-$ and $I_+$ two disjoint intervals which are made by removing the point $-1$
from an interval $I \ni -1$.
As $\nu(I_-)$ and $\nu(I_+)$ are separated by nonzero distance, by the split property (CN\ref{split}),
\begin{align*}
\sigma_\nu(\mathcal{A}(I_-)\vee\mathcal{A}(I_+))
&= \sigma_\nu^{I_-}(\mathcal{A}(I_-))\vee\sigma_\nu^{I_+}(\mathcal{A}(I_+)) \\
&= \mathcal{A}(\nu(I_-))\vee \mathcal{A}(\nu(I_+)) \\
&\simeq \mathcal{A}(\nu(I_-))\otimes \mathcal{A}(\nu(I_+)) \\
&=\sigma_\nu^{I_-}(\mathcal{A}(I_-))\otimes\sigma_\nu^{I_+}(\mathcal{A}(I_+)) \\
&\simeq \mathcal{A}(I_-)\otimes \mathcal{A}(I_+),
\end{align*}
where $\simeq$ means unitary equivalence mapping the algebra on the left (respectively right) to the algebra on the left
(respectively right).
Let us assume that $\sigma$ extended to a DHR representation.
Then it would be normal on $\A(I)$, and hence
$\sigma_\nu(\mathcal{A}(I_-)\vee\mathcal{A}(I_+))
\simeq \mathcal{A}(I_-)\vee \mathcal{A}(I_+)$.
But we know from \cite[page 292, Example b)]{Buchholz74}
that $\mathcal{A}(I_-)\vee\mathcal{A}(I_+)$ is not isomorphic to $\mathcal{A}(I_-)\otimes\mathcal{A}(I_+)$,
so this is a contradiction, and we conclude that $\sigma_\nu$ does not extend to a DHR representation.

Let us next show the unitary equivalence between $\sigma_{\nu_1}$ and $\sigma_{\nu_2}$.
This follows because $\nu_2 \circ \nu_1^{-1}$ is a diffeomorphism
from $I_{\nu_1}$ to $I_{\nu_2}$, which can be extended to $\nu_{1,2} \in \diff(S^1)$.
Now $U(\nu_{1,2})$ intertwines $\sigma_{\nu_1}$ and $\sigma_{\nu_2}$ because
by definition, on $\A(I)$ with $I \in {\overline{\I}_\RR}$,
\[
 \sigma_{\nu_2} = \Ad U(\nu_2^I) = \Ad U(\nu_{1,2})U(\nu_1^I) = \Ad U(\nu_{1,2}) \circ \sigma_{\nu1},
\]
where we used the fact that $\nu_{1,2}\circ \nu_1^I = \nu_2^I$ when
restricted to $I$.

With the result of the previous paragraph, in order to show local $\widetilde{\D^s(S^1)}$-covariance,
we can take a specific $\nu$. We take $\nu$ as the square root map, namely,
$\nu(e^{i\theta}) = e^{i\theta/2}$, where $\theta \in [-\pi, \pi)$.
Note that $\D^s(S^1)$ contains a copy of the $2$-cover $\D^s(S^1)^{(2)}$ of $\D^s(S^1)$,
and the square root map $\nu$ locally intertwines the action of $\D^s(S^1)^{(2)}$ on
the half-circle $[-\frac\pi 2, \frac \pi 2)$ and the action of $\D^s(S^1)$ on $S^1$
(see \cite[Section 2]{LX04}, \cite[Section 3.1 \S3]{WeinerThesis}).
Therefore, for $\g \in \widetilde{\D^s(S^1)}$ and $I \in {\overline{\I}_\RR}$ as in the definition of local covariance
and with the quotient map $q^{(2)}: \widetilde{\D^s(S^1)}\to\D^s(S^1)^{(2)}$,
$\Ad U(q^{(2)}(\g))\circ \sigma_{\nu} = \sigma_{\nu}\circ \Ad U(\g)$.
Namely, $\sigma_\nu$ is locally $\widetilde{\D^s(S^1)}$-covariant.
\end{proof}
As $\sigma_\nu$ is not irreducible, the implementation of $\widetilde{\D^s(S^1)}$ is not unique.
It is also possible to use the $n$-th root map in such a way that the covariance is
implemented by the restriction of $U$ to $\D^s(S^1)^{(n)}$.
The representation $\g \mapsto U(q^{(2)}(\g))$ used here has positive energy
by \cite[Corollary 3.6]{Weiner06} (applied to the vacuum representation).


\subsection{M\"obius covariance implies DHR}\label{DHR}
In Theorem \ref{th:type3} we saw that there is a proper, locally $\widetilde{\D^s(S^1)}$-covariant soliton
on any conformal net $\A$. In particular, it is locally $\uMob$-covariant.
In contrast, here we show that there is no locally $\mob$-covariant proper soliton.
The key is that $\Ad U^\sigma(R_{2\pi})$ is trivial in this case.
This is essentially contained in \cite[Lemma 2.6]{CHKLX15} and \cite[Proposition 19]{CKL08},
but we give a direct proof in our present setting.
\begin{proposition}\label{pr:mobDHR}
 Let $\s$ be a locally $\mob$-covariant soliton of a conformal net $\A$ with a representation $U^\s$
 of $\mob$.
 Then $\s$ extends to a DHR representation of $\A$.
\end{proposition}
\begin{proof}
 To extend $\s$ to $\A(I)$ where $I$ contains $-1$, we choose a rotation $R_\theta$
 such that $R_\theta I$ does not contain $-1$.
 We define $\s^I := \Ad U^\sigma(R_{-\theta})\circ \s \circ \Ad U(R_\theta)$,
 and show that this is well-defined.
 Indeed, assume that there are two such $0 < \theta_1, \theta_2 < 2\pi$.
 Then, either the rotation from $0$ to $\theta_1-\theta_2$ or the rotation $0$ to $\theta_1 - \theta_2 + 2\pi$
 brings $R_{\theta_2}I$ to $R_{\theta_1}I$ inside $S^1 \setminus \{-1\}$.
 As $R_{2\pi}$ is a scalar, without the loss of generality,
 we may consider the case of $\theta_1 - \theta_2$.
 By the assumption of locally $\mob$-covariance,
 \begin{align*}
  \Ad U^\sigma(R_{-{\theta_1}})\circ \s \circ \Ad U(R_{\theta_1})
  &= \Ad U^\sigma(R_{-{\theta_2}})U^\sigma(R_{-{\theta_1+\theta_2}})\circ \s \circ \Ad U(R_{\theta_1-\theta_2})U(R_{\theta_2}) \\
  &= \Ad U^\sigma(R_{-{\theta_2}})\circ \s \circ \Ad U(R_{\theta_2}),
 \end{align*}
 hence the definition does not depend on $\theta$ under the condition above.
 The compatibility condition follows from this, because if $\theta$ can be used for a larger interval,
 it works also for a smaller interval.
\end{proof}

There are also many irreducible locally $\uMob$-covariant solitons:
consider an inclusion of conformal nets $\A\subset \B$, take
a certain DHR representation $\rho$ of $\A$ and its $\a$-induction to $\B$.
For a generic $\rho$, the $\a$-induction is not a DHR representation
but just a soliton, and often such soliton can be decomposed into irreducible ones.
For concrete examples, see e.g.\! \cite[Section 5]{CHKLX15}
(this is due to Sebastiano Carpi, we thank him for his comments).

\section{Applications to infinite-dimensional groups}\label{examples}

Some conformal nets arise from a particular representation (often called the ``vacuum representation'')
of infinite-dimensional groups. Our construction in Section \ref{nonsmoothsoliton}
gives rise to a new class of representation of certain subgroups of them.

\subsection{The Virasoro net and representations of \texorpdfstring{$B_0$}{B0}}\label{sec:B0}

The Virasoro net with central charge $c$ is the conformal net induced by
the stress-energy tensor $T_{c,0}$ in the vacuum representation $\mathcal{H}(c,0)$ of the Virasoro algebra $\vir$:
\[
 \vir_c(I)=\{e^{iT_{c,0}(f)} : f\in C^{\infty}(S^1), \text{real-valued}, \text{supp}f\subset I\}^{\prime\prime}.
\]
With the lowest weight vector $\Omega_c$ and the unitary projective representation of $U_c$
of $\diff(S^1)$, $(\vir_c, U_c, \Omega_c)$ is a conformal net, see \cite[Section2.4]{Carpi04}.

The solitons in Section \ref{nonsmoothsoliton} give rise to positive energy representations of $B_0$
which do not extend to positive-energy representations of $\diff(S^1)$.
Let $\nu\in\diff(S^1,-1)$.
For any $\g \in B_0$, $\nu\circ \g\circ \nu^{-1}$ is $C^1$,
as the discontinuity of the first derivative of $\nu$ at the point of infinity gets cancelled.
We set
\begin{align*}
\alpha_\nu: B_0&\rightarrow \psone(S^1)\\
\g&\mapsto \nu\circ \g \circ \nu^{-1}.
\end{align*}
Clearly $\alpha_\nu$ is an homomorphism of $B_0$ into $\psone(S^1)$,
and $U_c$ extends to $\psone(S^1)$ by Proposition \ref{pr:psone} (although we do not know continuity).
We construct a projective unitary representation $U_c^\nu$ of $B_0$ by
\begin{align}\label{eq:repB0}
U_c^\nu(g)\coloneqq (U_c\circ \alpha_\nu)(g).
\end{align}
\begin{proposition}\label{pr:diffr}
Let $\nu\in\diff(S^1,-1)$ and assume that $r(\nu) \neq 1$.
The representation $U_c^\nu$ defined in \eqref{eq:repB0}
is not a restriction of any positive-energy representation of $\diff(S^1)$.
In addition, $U_c^{\nu_1}\simeq U_c^{\nu_2}$ if and only if $r(\nu_1)=r(\nu_2)$.
\end{proposition}
\begin{proof}
 The representation $U_c^\nu$ is irreducible, since $\bigvee_{I\in {\overline{\I}_\RR}} \A(I) = \B(\H)$
 by (CN\ref{cn:irreducibility}) and the left-hand side is generated by
 $\{U_c(\gamma):\supp \gamma \Subset \RR\} = \{U_c^\nu(\gamma):\supp \gamma \Subset \RR\}$.
 Suppose that $U_c^\nu$ were a restriction of a positive-energy representation of $\diff(S^1)$.
 By irreducibility, it would have to be a $(c',h')$-representation of Section \ref{smooth}.

 Any $(c',h')$-representation $U_{c',h'}$  is a projective representation of $\diff(S^1)$, in particular,
 the $2\pi$-rotation is a scalar.
 Furthermore, it holds\footnote{This equation holds including phase: 
 $\Ad U_{c',h'}(R_\theta)$ does not give phase by Proposition \ref{pr:covariance},
 where the phase $\beta(R_\theta, f)$ vanishes since $R_\theta$ is a M\"obius transformation,
 and exponentials $\Exp(f)$ generate the whole group $\diff(S^1)$, since $\diff(S^1)$ is algebraically simple \cite{Mather74}.}
 that $\Ad U_{c',h'}(R_\theta)(U_c^\nu(\gamma)) = U_c^\nu(R_\theta\circ\gamma\circ R_{-\theta})$
 as long as $\theta$ is small enough so that $R_\theta\, \supp \gamma \subset S^1\setminus\{-1\}$,
 because $U_{c',h'}$ is an extension of the projective representation $U_c^\nu$.
 Therefore, the restriction of it to $\mob$ makes
 the soliton $\sigma_\nu$ of $\vir_c$ locally $\mob$-covariant.
 By Proposition \ref{pr:mobDHR}, $\sigma_\nu$ would extend to a DHR representation,
 but  this contradicts with Theorem \ref{th:typeIsol}.
 This shows that $U_c^\nu$ does not extend to any $(c',h')$-representation.
 
 The claim about unitary equivalence also follows from Theorem \ref{th:typeIsol},
 by passing to $\sigma_{\nu_1}$ and $\sigma_{\nu_2}$.
\end{proof}

In Section \ref{uone}, we show that the representations $U_n$, where $n$ is a positive integer, extend to
to $\D^s(S^1), s>2$ by continuity.  As $\D^s(S^1)$ includes $\psone(S^1)$ if $2<s<5/2$, the representations $U^\nu_n$ are strongly continuous when $n\in\ZZ_+$ and
the solitons $\s_\nu$ are continuously covariant with respect to $B_0$, see Proposition \ref{pr:continuityu}.

\subsection{Loop group nets and representations of \texorpdfstring{$\Lambda G$}{Lambda G}}
We review the loop group nets, following \cite[Section III.2]{KoesterThesis}.
These results are based on \cite{GW84, Kac90, GF93, Toledano-Laredo99-1}.

Let $G$ be a simple, compact, connected and simply connected Lie group.
The group of smooth maps from $S^1$ to $G$ is denoted by $LG$.   
With $\Lambda G$ we denote the group of smooth maps $\mathbb{R}\rightarrow G$ with compact support.
This group $\Lambda G$ is identified through Cayley transform
with the subgroup of $LG$ of elements whose support does not contain $-1$.

Lie algebra $L\mathfrak{g}$ consisting of smooth maps from $S^1$ to $\mathfrak{g}$
with the pointwise operation is called the loop algebra, and it is the Lie algebra of
the loop group $LG$ in the infinite-dimensional sense (see \cite{PS86}).
A 2-cocycle of $L\mathfrak{g}$ is a bilinear form $\omega:L\mathfrak{g}\times L\mathfrak{g}\rightarrow \RR$ such that
\begin{align*}
\omega([x,y],z)+\omega([y,z],x)+\omega([z,x],y)=0.
\end{align*}
With such an $\omega$ it is possible to construct a central extension of $\widetilde{L\mathfrak{g}}$
of $L\mathfrak{g}$ by a one-dimensional centre: as a vector space,
$\widetilde{L\mathfrak{g}}=L\mathfrak{g}\oplus\RR$ with bracket
\begin{align*}
[(x,a),(y,b)]\coloneqq ([x,y],\omega(x,y))
\end{align*}
with $x,y\in L\mathfrak{g}$ and $a,b\in\RR$.
If $\mathfrak{g}$ is simple, every continuous $G$-invariant 2-cocycle $\omega$ has the form
\begin{align*}
\omega(x,y)=\frac{1}{2\pi}\int_0^{2\pi}\langle x(\theta),y'(\theta)\rangle d\theta
\end{align*}
where $\langle\cdot,\cdot\rangle$ is a symmetric invariant form on $\mathfrak{g}$,
which is unique up to a scalar.

The constant functions in $L\mathfrak g$ can be identified with $\mathfrak g$.
Let us fix a basis $\{j_0^a\}$ in $\mathfrak g$.
The complexification $\widetilde{L\mathfrak{g}}_\CC$ of $\widetilde{L\mathfrak{g}}$ contains
functions $e^{in\theta}$ multiplied with an element $j^a_0$ of the basis in $\mathfrak g$.
Let us denote these elements by $j^a_n$. They satisfy the commutation relation
\[
 [j_m^a, j_n^b] = if^{ab}_c j_{m+n}^c + \omega(j^a_0, j^b_0)m\delta_{m,-n}.
\]
One can construct the ``vacuum'' representations $\pi^G_{\ell,0}$, namely, a representation which contains
a vector $\Omega$ such that $\pi^G_{\ell,0}(j_n^a)\Omega = 0$ for $n \ge 0$.
The central element $(0,a)$ is represented by a scalar $\ell$, which is called the level.
One can introduce a scalar product on this representation space in such a way that
$\pi^G_{\ell,0}(j_n^a)^\dagger = \pi^G_{\ell,0}(j_{-n}^a)$ and $\|\Omega\| = 1$.
Such a scalar product is positive definite if and only if the level $\ell$ is positive integer.
Furthermore, for a real element $f \in \widetilde{L\mathfrak{g}}$,
$\pi^G_{\ell,0}(f)$ is an essentially self-adjoint operator on the domain
generated by $\Omega$ and $\{\pi^G_{\ell,0}(j_{-n}^a)\}$.

A projective unitary representation $V$ of a group on a Hilbert space $\H$ is a map
from the group in  $\U(\H)$ such that $V(g)V(h)=c(g,h)V(gh)$ for some scalar $c(g,h)$.
The vacuum representation of $L\mathfrak g$ at level $\ell$
integrates to a projective unitary representations of $LG$:
for a real element $f \in \widetilde{L\mathfrak{g}}$,
it holds that $V_{\ell,0}(\Exp f) = e^{i\pi^G_{\ell,0}(f)}$ up to a scalar.
Furthermore, there is a projective unitary representation $U$ of $\diff(S^1)$
such that $U(\g)V_{\ell,0}(g)U(\g)^* = V_{\ell,0}(g\circ \g^{-1})$.

A projective unitary representation $V$ of $LG$ on $\H$ is said to have positive energy
if there exists a strongly continuous unitary representation $U$ of the rotation group $\mathbb{T}$
on the same Hilbert space with positive generator such that
\begin{align*}
U(R_\theta)V(g)U(R_\theta)^*=V(g_\theta),
\end{align*} 
where $g_\theta(e^{i\varphi})\coloneqq g(e^{i(\varphi-\theta)})$.
Correspondingly, we say that a projective unitary representation $V$ of $\Lambda G$ has positive energy
if there exists a strongly continuous unitary representation $U$ of the translation group $\RR$
such that
\begin{align*}
U(\tau_t)V(g)U(\tau_t)^* = V(g_t)
\end{align*}
where $g_t(s)=f(s-t)$.
We have \cite[Proposition 9.2.6]{PS86}:
\begin{proposition}
The restriction to $\Lambda G$ of a positive energy representation of $LG$ is a positive energy representation of $\Lambda G$.
\end{proposition}

Let $V^G_{\ell,0}$ be the vacuum representation of level $\ell$, which has positive energy
with respect to $U$.
With the family of von Neumann algebras 
\begin{align*}
\A_{G,\ell}(I)\coloneqq \left\lbrace V^G_{\ell,0}(g):\supp g\subset I\right\rbrace''
\end{align*} 
$(\A_{G,\ell}, U, \Omega)$ is a conformal net. They are called the {\bf loop group nets}
with group $G$ at level $\ell$.

\begin{proposition}
There exist irreducible positive energy representations of $\Lambda G$ which do not extend to positive energy representations of $LG$.
\end{proposition}
\begin{proof}
Fix a level $\ell$ and consider the conformal net $\A_{G,\ell}$.
Then we can construct a representation $V^{G,\nu}_{\ell,0}$ of $\Lambda G$ by
\begin{align*}
V^{G,\nu}_{\ell,0}:=\s_\nu\circ V^G_{\ell,0},
\end{align*}
where $\s_\nu$ is a proper soliton of the conformal net $\A_{G,\ell}$ with $\nu\in\diff(S^1,-1)$
as in Section \ref{nonsmoothsoliton}.
By Theorem \ref{th:positivitysol}, $V^{G,\nu}_{\ell,0}$ it has positive energy.

Suppose that $V^{G,\nu}_{\ell,0}$ were the restriction $V$ of a positive energy representation of $LG$.
Then $V^{G,\nu}_{\ell,0}$ would also be irreducible as a representation of $LG$ as in Proposition \ref{pr:diffr}.
Such $V$ must have positive energy by \cite[Theorem 7.4]{GW84}, namely,
there is a projective unitary representation $U^\nu$ of $\diff(S^1)$ whose restriction to the rotations
has positive generator and $U^\nu(\g)V(g)U^\nu(\g) = V(g\circ \g^{-1})$.
As the restriction of $V$ to $\Lambda G$ is $V^{G,\nu}_{\ell,0}$,
$U^\nu$ makes $V^{G,\nu}_{\ell,0}$ locally $\diff(S^1)$-covariant, especially it is locally $\mob$-covariant
(and not just locally $\uMob$-covariant, namely the $2\pi$-rotation is trivial).
Accordingly, the soliton $\s_\nu$ is also locally $\mob$-covariant.
By Proposition \ref{pr:mobDHR}, it should extend to a DHR representation of the net.
This contradicts with
Theorem \ref{th:typeIsol}. Therefore, $V^{G,\nu}_{\ell,0}$ does not arise from the restriction
of any positive energy representation of $LG$.
\end{proof}

The existence of such representation has been marked as an open problem in \cite[P.174, Remark]{PS86}.
The type III soliton from Section \ref{nonsmoothsoliton} gives another such representation, since
it cannot be locally $\mob$-covariant as we saw in Theorem \ref{th:type3} and Proposition \ref{pr:mobDHR}.

\section{Sobolev diffeomorphism covariance of the \texorpdfstring{$\mathrm{U}(1)$}{U(1)}-current net}\label{uone}
Here we take the $\mathrm{U}(1)$-current net where there is the criterion of Shale-Steinspring
to determine whether an automorphism of the algebra can be unitarily implemented.
Indeed, following the strategy of \cite{Vromen13},
we show that $\D^s(S^1)$-diffeomorphisms are implemented with $s > 2$, and this group
includes $\psone(S^1)$ and $\diff^3(S^1)$ and also $C^1$-piecewise M\"obius group \cite{WeinerThesis}.

Let $\K$ be a complex Hilbert space with the scalar product $\<\cdot,\cdot\>$
The $C^*$-algebra generated by the operators $W(f)$, $f\in\K$,
satisfying the relations $W(f)W(g)=e^{-i\Im \<f,g\>/2}W(f+g)$ and $W(0)=\1$ is called the CCR algebra.
There is a representation of the CCR algebra on the symmetric Fock space
$\G_+(\K) = \bigoplus \mathrm{Sym}\K^{\otimes n}$
with the Fock vacuum $\Omega$ and $W(f)\Omega = \sum_n \bigoplus_j \frac1{n!} f^{\otimes n}$.
We denote the Weyl operators in this representation with the same $W(f)$.
If $f\in\K$ and $A$ is a {\it real linear}, invertible operator on $\K$
which preserves the symplectic bilinear form $\im\<\cdot,\cdot\>$,
then the map $W(f)\mapsto W(Af)$ is a *-automorphism of the CCR algebra.
Such a *-automorphism is unitary implemented in the Fock representation
if and only if
$\frac12 J[A,J]$ is an Hilbert-Schmidt operator, where $J$ is the multiplication by the imaginary unit \cite[Theorem 4.1]{Shale62}.
We partly use the conventions of \cite[Section 5.3]{Ottesen95}.
For any $f\in\K$ the Weyl operators $W(f)$ on $\Gamma_+(\K)$
satisfy strong continuity:
if $f_n\rightarrow f$ in $\K$ then $\Vert(W(f_n)-W(f))\xi\Vert\rightarrow 0$ for every $\xi\in\Gamma_+(\K)$.

Let $C^{\infty}(S^1,\RR)$ be the space of real-valued smooth function on $S^1$.
We define a seminorm on it by
\begin{align}\label{eq:uonenorm}
\Vert f\Vert \coloneqq\sum_{k\in\mathbb{N}}k|\hat{f}_k|^2.
\end{align}
We introduce a complex structure on $C^{\infty}(S^1,\RR)$ by means of the operator $J$:
\begin{align*}
J\left(\sum_{k\in\mathbb{Z}\setminus\lbrace 0\rbrace}f_k e_k\right)\coloneqq \sum_{k\in\mathbb{N}}(if_k)e_k+\sum_{k\in\mathbb{N}}(-if_{-k})e_{-k},
\end{align*}
where $e_k(e^{i\theta})\coloneqq e^{ik\theta}$.
The space $C^{\infty}(S^1,\RR)$ quotiented by the null space with respect to the norm
${\lbrace f\in C^{\infty}(S^1,\RR):\Vert f\Vert=0\rbrace}$
is equipped with the complex structure $J$.
With $J$ as the imaginary unit, the quotiented space becomes the complex Hilbert space $\H_1$.
This space admits the irreducible unitary representation $U_1$ of $\mob = \psl2r$ with lowest weight 1.
The action of $\mob$ on $C^{\infty}(S^1,\RR)$
\begin{align*}
U_1(\g)(f)\coloneqq f\circ\g^{-1}
\end{align*} 
extends to $\H_1$.
For a function $f\in C^{\infty}(S^1,\mathbb{R})$ we denote with $[f]$ its image in $\H_1$.
The seminorm is induced by the complex scalar product
\begin{align*}
\<f,g\>\coloneqq \frac{1}{2}\sum_{k\in\mathbb{N}}k(\overline{\hat{f}_k}\hat{g}_k+\hat{f}_k\overline{\hat{g}_k}).
\end{align*}
The M\"obius group acts on $\Gamma_+(\H_1)$
via the second quantization, and we denote it by $U(\g)\coloneqq \G_+(U_{1}(\g))$.
The adjoint action of $U(\g)$ on the Weyl operators is particularly simple:
\begin{align*}
\Ad U(\g)W([f])=W(U_1(\g)[f]).
\end{align*}

The family of von Neumann algebras
\begin{align*}
\mathcal{A}_{U(1)}(I)\coloneqq \lbrace W([f]):f\in C^{\infty}(S^1,\mathbb{R}),\text{ }\supp(f)\subset I\rbrace''
\end{align*}
with the Fock vacuum vector $\Omega \in \Gamma_+(\H_1)$ and
the representation $U$ is a M\"obius covariant net \cite{GLW98}.
The representation $U$ of $\psl2r$ can be extended to a projective representation $U$ of
$\diff(S^1)$ in such a way that $\A_{U(1)}$ is actually a conformal net, see \cite[Theorem 9.3.1]{PS86}.
We show that $U$ can be extended to $\D^s(S^1), s > 2$.

In the following, elements in the universal covering $\widetilde{\D^s(S^1)}$ are considered as
maps $\tilde \g$ from $\RR \to \RR$ such that $\tilde \g(\theta + 2\pi) = \tilde \g(\theta) + 2\pi$.
\begin{lemma}\label{lm:lambda}
Let $\g\in\D^s(S^1)$, $s>3/2$, the image of $\tilde{\g}\in\widetilde{\D^s(S^1)}$ through the covering map and $\lambda_{m,n}\coloneqq \frac{1}{2\pi}\int_0^{2\pi} e^{-im\theta}e^{in\tilde{\g}(\theta)}d\theta$,
where $m,n$ are either $m<0, n>0$ or $m>0, n<0$. Then there exists $C_{s,\g}\ge 0$ such that 
\begin{align*}
|\lambda_{m,n}|\le \frac{C_{s,\g}}{\left(|m|+|n|\right)^{s-1}}.
\end{align*}
\end{lemma}
\begin{proof}
As in the proof of \cite[Proposition 5.3]{Segal81}, consider the path $\tilde\g_t$ in $\widetilde{\D^s(S^1)}$:
\[
 [0,1]\ni t\mapsto \tilde{\g}_t\coloneqq t\tilde{\g}+(1-t)\id\in\widetilde{\D^s}(S^1).
\]
This is indeed a path in $\widetilde{\D^s(S^1)}$, because $\tilde \g_t'(\theta) > 0$.

For $0\le t \le 1,$ we have $\left(\tilde{\g}_t^{-1}\right)'\in H^{s-1}(S^1)$ by \cite[Theorem B.2(ii), Lemma B.1]{IKT13}.
From the definition of the norm $\|f\|_{s-1} = \left(\sum_k (1+k^2)^{s-1}|\hat f_k|^2\right)^\frac12$
where $\hat f_k$ is the $k$-th Fourier component, it follows that
\[
 \left|\widehat{\left(\tilde{\g}_t^{-1}\right)'}_{\mp(|m|+|n|)}\right|
 \le \frac{\left\Vert\left(\tilde{\g}_t^{-1}\right)'\right\Vert_{s-1}}{\left(|m|+|n|\right)^{s-1}}
 \le \frac{2\pi C_{s,\g}}{\left(|m|+|n|\right)^{s-1}},
\]
where $\sup_t \left\{\left\Vert\left(\tilde{\g}^{-1}_t\right)'\right\Vert\right\}
=: 2\pi C_{s,\g}$ which is finite, because $t\mapsto \tilde{\g}_t$ is continuous in $\D^s(S^1)$
and their first derivatives are uniformly separated from $0$.

By setting $t=\frac{|n|}{|m|+|n|}$, with $+$ corresponding to the case $m<0,n>0$ and $-$ corresponding to $m>0,n<0$,
we have
\begin{align*}
\lambda_{m,n}=\frac{1}{2\pi}\int_0^{2\pi} e^{\pm i(|n|+|m|)\tilde{\g}_t(\theta)}d\theta=\frac{1}{2\pi}\int_0^{2\pi} e^{\pm i(|n|+|m|)\varphi}\left(\tilde{\g}_t^{-1}\right)'(\varphi)d\varphi
= \frac1{2\pi}\widehat{\left(\tilde{\g}_t^{-1}\right)'}_{\mp(|m|+|n|)},
\end{align*}
therefore, $|\lambda_{m,n}| \le \frac{C_{s,\g}}{\left(|m|+|n|\right)^{s-1}}$ as desired.
\end{proof}

Note that the map $V(\g)[f]\coloneqq [f\circ\g^{-1}]$ for $\g\in\D^s(S^1)$ is well-defined,
because the kernel of $[\cdot]$ is the constant functions and they remain constant after
composition by $\g^{-1}$.

In order to estimate the Hilbert-Schmidt norm of $A_{V(\g)}\coloneqq \frac{1}{2}J[V(\g),J]$,
note that $A_{V(\g)}$ is anti-complex linear \cite[Section 5.3]{Ottesen95}, namely,
$A_{V(\g)}J = -JA_{V(\g)}$.
Therefore, its Hilbert-Schmidt norm on the complex Hilbert space
$\H_1$ is just the half of its Hilbert-Schmidt norm on $\H_1$ as a real Hilbert space.
If we put the norm defined by \eqref{eq:uonenorm} on $C^\infty(S^1, \CC)$,
we obtain a complex Hilbert space $\H_1^\CC$ which is naturally isomorphic to the direct sum
of two copies of $\H_1$, where the complex structure is given by $J$ above.
This space $\H_1^\CC$ has the basis $\left\{\frac1{\sqrt k}e_k\right\}_{k \in \ZZ, k\ne 0}$,
where $e_m(\theta) = e^{im\theta}$,
and the operator $A_{V(\g)}$ can be extended diagonally and its Hilbert-Schmidt norm
on $\H_1$ as a real linear operator is the same as its Hilbert-Schmidt norm
on $\H_1^\CC$.
\begin{proposition}\label{pr:implementation}
Let $\g\in\D^s(S^1)$, $s>2$.
Then there is a unitary operator $U(\g)$ which implements the action
to the CCR algebra corresponding to the map $V(\g)$, namely,
$\Ad U(\gamma)(W([f])) = W([f\circ\gamma^{-1}])$, for $f\in C^{\infty}(S^1,\RR)$.
\end{proposition}
\begin{proof}
Let $f,g\in C^{\infty}(S^1,\RR)$. The (real) symplectic bilinear form
$\sigma([f],[g])\coloneqq \Im \<f,g\>$
can be written as follows:
\begin{align*}
\sigma([f],[g])=\frac{1}{4\pi}\int_0^{2\pi} f(e^{i\theta})g'(e^{i\theta})d\theta.
\end{align*}
As $\g\in\D^s(S^1)$, $s>2$, $\g$ is in $\diff^1(S^1)$ and
the map $V(\g)$ preserves the symplectic form $\sigma(\cdot,\cdot)$.

Following \cite[Theorem 24]{Vromen13}, we only need to show that the Hilbert-Schmidt norm of the operator
$[V(\g),J]$ is finite. As remarked above,
we can compute it on $\H_1^\CC$ with the basis $\left\{\frac1{\sqrt k}e_k\right\}_{k \in \ZZ, k\neq 0}$.
The scalar product $\left\<\frac1{\sqrt m}e_m, \frac12 J[V(\g),J]\frac1{\sqrt{n}}e_n\right\>$
vanishes when $m>0,n>0$ or $m<0,n<0$.
The remaining cases are $m<0, n>0$ and $m>0, n<0$ and
\begin{align*}
 \left|\left\<\frac1{\sqrt m}e_m, \frac12 J[V(\g),J]\frac1{\sqrt{n}}e_n\right\>\right|^2
&= \frac{|m|}{|n|} \left|\<e_m, V(\g)e_n\>_{L^2(S^1)}\right|^2 \\
&= \frac{|m|}{|n|} |\lambda_{m,n}|^2.
\end{align*}
With $A_{V(\g)} = \frac{1}{2}J[V(\g),J]$, by Lemma \ref{lm:lambda} we have
\begin{align*}
\frac14\Vert A_{V(\g)}\Vert_{\mathrm{HS}}^2=\sum_{m>0,n<0}\frac{|m|}{|n|}|\lambda_{m,n}|^2
\le \sum_{m>0,n<0}\frac{|m|}{|n|}\frac{C_{s,\g}^2}{\left(|m|+|n|\right)^{2(s-1)}}.
\end{align*}
Let $p\coloneqq |m|+|n|$, then
\begin{align*}
\sum_{m>0,n>0}\frac{m}{n\left(m+n\right)^{2(s-1)}}=\sum_{p>0}\frac{1}{p^{2(s-1)}}\sum_{n=1}^{p-1}\frac{p-n}{n}\le \sum_{p>0}\frac{p-1}{p^{2(s-1)}}\sum_{n=1}^{p-1}\frac{1}{n}\le \sum_{p>0}\frac{(p-1)\left(2+\log(p)\right)}{p^{2(s-1)}}
\end{align*}
which converges if $s>2$, therefore, $\Vert A_{V(\g)}\Vert_{\mathrm{HS}}^2 < \infty$.
\end{proof}

\begin{theorem}\label{th:repcontinuity}
The map $\alpha:\D^s(S^1)\rightarrow\Aut(\B(\G_+(\H_1)))$ such that $\g\mapsto \a_\g \coloneqq\Ad U(\g)$
is pointwise strongly continuous if $s>2$.
\end{theorem}
\begin{proof}
Let $f\in C^{\infty}(S^1,\RR)$ and $\lbrace\g_n\rbrace\subset\D^s(S^1)$ a sequence converging to $\g$ in $\D^s(S^1)$. Recall that $C^{\infty}(S^1,\RR)\subset H^s(S^1)$ for every $s$ and that if $f\in H^s(S^1)$, $s\ge 1/2$,
then $\Vert f\Vert\le\Vert f\Vert_s$, where $\Vert f\Vert \coloneqq\sum_{k\in\mathbb{N}}k|\hat{f}_k|^2$.
By Lemma \ref{lm:sobolevcomp}, the map $(f,\g)\mapsto f\circ\g^{-1}$ is continuous for $s>3/2$.
Using Proposition \ref{pr:implementation} and the strong continuity of the Weyl operators,
it follows that for $s>2$, the map $\alpha_{\g_n}(W([f]))\rightarrow\alpha_{\g}(W([f]))$, $f\in C^{\infty}(S^1,\RR)$.

Let $\W$ be the linear span of Weyl operators $W([f])$. By the previous paragraph, we have
$\lim_{n\rightarrow \infty} \Ad U(\gamma_n)(x) = \Ad U(\gamma)(x)$ in the strong topology for every $x\in \W$,
and $\W$ is dense in $\B(\G_+(\H_1))$ in the strong operator topology.
Now let $\{\xi_n\}\subset\Gamma_+(\H_1)$ be a dense sequence. Let $A\in B(\Gamma_+(\H_1))$.
By Kaplanski's density theorem we can choose a sequence $\{A_m\}\subset \W$ such that $A_m\rightarrow A$ strongly.
Thus we have for every $\xi_n$
\[
 \lim_{m\rightarrow \infty} \Ad U(\gamma)(A_m)\xi_n = \Ad U(\gamma)(A)\xi_n,
\]
i.e.\! $f_n(\gamma):= \Ad U(\gamma)(A)\xi_n$ is the pointwise limit of $f_{n,m}(\gamma):=\Ad U(\gamma)(A_m)\xi_n$.
Note that $\D^s(S^1)$ is a Baire space, since it is an open set of a complete metric space
\cite[Lemma B.2, cf.\! Corollary 2.1(ii)]{IKT13}.
By Baire-Osgood's theorem \cite[Theorem 11.20]{Carothers00}\cite{PW-Baire-Osgood} applied to the maps $f_{n,m}$ and $f_n$
from a Baire space $\D^s(S^1)$ into the Hilbert space $\G_+(\H_1)$,
we get that the set
\[
 D(f_n):=\{\gamma\in \D^s(S^1): f_n\text{\hspace{1mm} is not continuous in } \gamma\}
\]
is meager. Thus also $\bigcup_n D(f_n)$ is meager. It follows that $\D^s(S^1)\setminus \bigcup_n D(f_n)$
is nonempty and hence there is $\gamma_0 \in \D^s(S^1)$ for which all $f_n$ are continuous.
Since $\{\xi_n\}$ is dense,
\begin{align*}
\gamma\mapsto \Ad U(\gamma)(A)\xi\eqqcolon f^A_{\xi}(\g)
\end{align*}
is continuous at $\gamma_0$ for every $\xi\in\G_+(\H_1)$. 
Set $\g_1\coloneqq \g_0^{-1}\g$, then 
\begin{align*}
g^A_{\xi}(\g_1)\coloneqq \Ad U(\g_1)(A)\xi = \Ad [U(\g_0)^*U(\g)](A)\xi = U(\g_0^{-1})f^A_{U(\g_0)\xi}(\g)
\end{align*}
converges to $U(\g_0)^*f^A_{U(\g_0)\xi}(\g_0)$ as $\g \to \g_0$ for every $A\in\B(\G_+(\H_1))$ and for every $\xi\in\G_+(\H_1)$.
In other words, the map $\g_1 \mapsto \Ad U(\g_1)$ is pointwise continuous in the strong operator topology
at the identity $\id \in\D^s(S^1)$.

Since the map 
\begin{align*}
\gamma\mapsto \text \Ad U(\gamma)\in{\Aut(\B(\G_+(\H_1)))}
\end{align*}
 is a group homomorphism and is continuous at $\id$
 it is continuous for every $\g\in\D^s(S^1)$.
\end{proof}

As we saw in Lemma \ref{lm:smoothaction}, $\psone(S^1)\subset\D^s(S^1)$ if $s<5/2$.
\begin{corollary}
The $\mathrm{U}(1)$-current net $\A_{U(1)}$ is continuously $\D^s(S^1)$-covariant, $s>2$,
and in particular is $\psone(S^1)$-covariant.
\end{corollary}
\begin{proof}
The proof is the same as in \cite[Proposition 4.1]{CDIT18+}.
\end{proof}
Note that this is stronger than the general result, Proposition \ref{pr:psone},
as here we have the continuity of $\psone(S^1)$-action as a subgroup of $\D^s(S^1)$.

\begin{corollary}\label{cr:diffext}
The projective unitary representation $U$ of $\diff(S^1)$ on $\G_+(\H_1)$
extends continuously to $\D^s(S^1), s>2$.
\end{corollary}
\begin{proof}
 The map $\g \mapsto \Ad U(\g)$ is continuous by Theorem \ref{th:repcontinuity},
 and this is equivalent to the continuity of $\g \mapsto U(\g)$ in $\U(\G_+(\H_1))/\mathbb{T}$.
\end{proof}

\begin{corollary}\label{cr:virasorocovariance}
The Virasoro net $\vir_1$ with central charge $c=1$ is continuously $\D^s(S^1)$-covariant, $s>2$,
and in particular is $\psone(S^1)$-covariant.
\end{corollary}
\begin{proof}
Let $\vir_1$ the Virasoro net of central charge $c=1$,
where
\[
 \vir_1(I) = \{U(\g):\g \in \diff(S^1),\; \supp \g \subset I\}'',
\]
which is a subnet of $\A_{\uone}$.
The subspace $\H_{\vir_1}\coloneqq \overline{\bigcup_{I}\vir_1(I)\Omega}$ is invariant
for $U(\g), \g \in \D^s(S^1)$ by Corollary \ref{cr:diffext},
hence the representation $U$ restricts to $\H_{\vir_1}$, and the covariance follows.
\end{proof}

Let $U_{1,0}$ the irreducible positive energy projective unitary representation of $\diff(S^1)$
with central charge $1$ and lowest weight $0$. The finite tensor product $U_n\coloneqq \bigotimes_n U_{1,0}$,
is a positive energy projective representation of $\diff(S^1)$ which contains $U_{n,0}$ as a subrepresentation.
By Corollary \ref{cr:virasorocovariance}, all the representations $U_{n,0}$ of $\diff(S^1)$
extend to $\D^s(S^1), s > 2$. This is a partial improvement of the results of \cite{CDIT18+},
where all $U_{c,h}$ were extended to $\D^s(S^1), s > 3$.

We now show that for these conformal nets the representations of $B_0$ constructed in Section \ref{sec:B0} are strongly continuous.

\begin{lemma}\label{lm:continuityalpha}
 Let $\mathring{\g}\in B_0$, $\mathring{\nu}\in\diff(S^1,-1)$ and $2<s<5/2$.
 The homomorphism $\a_{\mathring{\nu}}:B_0\longrightarrow \D^s(S^1)$,
 $\mathring{\g}\mapsto \a_{\mathring{\nu}}(\mathring{\g})\coloneqq \mathring{\nu}\circ \mathring{\g}\circ\mathring{\nu}^{-1}$,
 where $B_0$ is equipped with the $C^\infty$-topology, is continuous.
\end{lemma}
\begin{proof}
 Let $\lbrace \mathring{\g_n}\rbrace\subset B_0$ be a sequence converging to $\mathring{\g}\in B_0$
 with respect to the $C^\infty$-topology. We denote with $\nu$ the lift to $\widetilde{\diff^0(S^1)}$ of $\mathring{\nu}$
 and with $\g_n$ and $\g$ the lift to $\widetilde{B_0}$ of $\mathring{\g_n}$ and $\mathring{\g}$, respectively. 
 We use the same strategy of Lemma \ref{lm:nonsmoothaction}. Namely, the convergence
 $\nu\circ\g_n\circ\nu^{-1} \to \nu\circ\g\circ\nu^{-1}$ in the $L^1(S^1)$-topology (actually, even in the uniform topology)
 is straightforward. Then, by
  \begin{align*}
   \left\vert\left(\widehat{\nu\circ \g_n\circ\nu^{-1}}\right)_k\right\vert
   \leq\frac{\mathrm{Var}\left(\left(\nu\circ \g_n\circ\nu^{-1}\right)''\right)}{k^3}
\end{align*}
 it is sufficient to show that the right-hand side is uniformly bounded in $n$.
 The second derivative of $\nu\circ \g_n\circ\nu^{-1}$ is
 \begin{align}\label{eq:variation}
 \frac{d^2}{d\theta^2}\left(\nu\circ \g_n\circ\nu^{-1}\right)(\theta)
 &=\nu''(\g_n(\nu^{-1}(\theta)))\g_n'(\nu^{-1}(\theta))^2\frac{1}{\nu'(\nu^{-1}(\theta))^2} \nonumber\\
 &\qquad +\nu'(\g_n(\nu^{-1}(\theta)))\g_n''(\nu^{-1}(\theta))\frac{1}{\nu'(\nu^{-1}(\theta))^2}\\
 &\qquad -\nu'(\g_n(\nu^{-1}(\theta)))\g_n'(\nu^{-1}(\theta))\frac{\nu''(\nu^{-1}(\theta))}{\nu'(\nu^{-1}(\theta))^3}. \nonumber
 \end{align}
 To evaluate its total variation, we use the following facts:
 for every pair of functions $f_1,f_2$ with bounded variation, it holds \cite[Theorem 3.7]{Pausinger15} that
 \begin{align*}
  \mathrm{Var}(f_1\cdot f_2)
  &\leq\Vert f_1\Vert_\infty\mathrm{Var}(f_2)+\Vert f_2\Vert_\infty\mathrm{Var}(f_1)+3\mathrm{Var}(f_1)\mathrm{Var}(f_2) \\
   \mathrm{Var}(f_1\circ f_2)&\leq L_{f_1}\mathrm{Var}(f_2),
 \end{align*}
 where $f_1$ is Lipschitz and $L_{f_1}$ is the Lipschitz constant of $f_1$.
 Now, the total variations of the second and the third terms are uniformly bounded in $n$
 since $L_{\g_n^{(k)}}$ are uniformly bounded in $n$.
 As for the first term, we have $\mathrm{Var}(\nu''\circ \g_n\circ\nu^{-1})
 \leq 2\pi\left\Vert\left(\nu''\circ \g_n\circ\nu^{-1}\right)'\right\Vert_{L^\infty(0,2\pi)}+|\nu''(2\pi)-\nu''(0)|$,
 and this is again uniformly bounded since $\nu''$ has a bounded derivative on the open interval $(0,2\pi)$
 and $L_{\g_n^{(k)}}$ are uniformly bounded in $n$.
\end{proof}

\begin{proposition}\label{pr:continuityu}
Let $2<s<5/2$ and $\g\in\diff(S^1,-1)$. The map $U_\g\coloneqq U\circ\alpha_\g$ is a strongly continuous unitary projective representation of $B_0$ when $U=U_{n,0}$, $n\in\ZZ_+$, or $U$ is as in Corollary \ref{cr:diffext}.

Let $\A$ be the $\mathrm{U}(1)$-current net or the Virasoro net $\A_{\vir_c}$ with $c\in\ZZ_+$ and $\g\in\diff(S^1,-1)$.
Every soliton $\sigma_\g$ of $\A$ as in Section \ref{nonsmoothsoliton} is continuously $B_0$-covariant
with respect to the representation $U_\g$.
\end{proposition}
\begin{proof}
This is clear from Corollary \ref{cr:diffext} and Lemma \ref{lm:continuityalpha}.
\end{proof}

\section{Outlook}\label{outlook}
Let us collect some open problems.
\begin{itemize}
 \item There appears to be no known soliton which is not dilation-covariant.
Is dilation covariance automatic in solitons? This is not obvious,
because we cannot implement dilations by cutting the generators as we did
for translations.
 \item Is it possible to classify all solitons for some specific conformal nets?
For example, for Virasoro nets any such soliton should give rise to
a representation of the group of the diffeomorphisms of $\RR$ with compact support.
Yet, the lack of any compact subgroup makes it difficult to classify such
representations.
 \item As the action $\g_*(f)$ of a diffeomorphism on vector fields involves
 the derivative of $\g$, it may decrease the regularity, especially, it may have discontinuous derivative.
 On the other hand, if $\g$ and $\Exp(f)$ are implementable, so is $\Exp(\g_*(f))$
 by $U(\g)U(\Exp(f))U(\g)^*$,
 hence implementability of $\Exp(g)$ is not directly related with the regularity of $g$.
 What is the precise relationship? Is $\dom(L_0)$ a core for such $T(g)$?
 \item Which is the smallest $s>0$ for which conformal nets are $\D^s(S^1)$-covariant?
 Does $s$ depend on the net? For which $s$ do $(c,h)$-representations of $\diff(S^1)$ extend to
 $\D^s(S^1)$?
\end{itemize}

\appendix
\section{Piecewise smooth \texorpdfstring{$C^1$}{C1}-diffeomorphisms}\label{c1pws}
Here we show that any $\g \in \psone(S^1)$ is implementable in any conformal net.
The strategy is due to Andr\'e Henriques. We thank him for permitting us to include this in the present paper.

An element $\g \in \psone(S^1)$ has only finitely many nonsmooth points, hence if we show that
any $\g$ with one nonsmooth point is implemented, the thesis follows by composing such elements
finitely many times. Furthermore, by composing with rotations and dilations, we may assume that the nonsmooth
point is $-1$ and $\g(-1) = -1, \g'(-1) = 1$.

If $\gamma\in\psone(S^1)$, let $\tilde{\gamma}$ be a lift of $\gamma$ to the universal covering $\widetilde{\diff^1(S^1)}$.
As we observed in Section \ref{nonsmoothsoliton},
there exists an open interval $I$ of $S^1$ which contains $-1$ and $\g_{I_{-}},\g_{I_{+}}\in\diff(S^1)$ such that $\g$ agrees with $\g_{I_{-}}$ in $I_-$ and with $\g_{I_{+}}$ in $I_+$, where $I_-$ and $I_+$ are the connected components of $I\setminus\lbrace -1\rbrace$. Denote the derivative of $\g$ from the right and from the left by
\begin{align*}
\partial_\pm \g(-1)\coloneqq \lim_{\theta \to \pi_\pm}\frac{\widetilde{\g_{I_{\pm}}}(\theta)-\widetilde{\g_{I_{\pm}}}(\pi)}{\theta-\pi}.
\end{align*}
For $f\in C^{\infty}(S^1,\RR)$ and $\g\in\diff(S^1)$ we define
\begin{align*}
f^{(k)}(e^{i\theta})\coloneqq \frac{d^k}{d\theta^k}f(e^{i\theta})
\end{align*}
and
\begin{align*}
\g^{(k)}(e^{i\theta})\coloneqq \frac{d^k}{d\theta^k}\tilde{\g}(\theta)
\end{align*}
where $\tilde{\g}$ is the lift of $\g$ in $\widetilde{\diff(S^1)}$
and $\widetilde{\diff(S^1)}$ is identified with the group of maps $\tilde \g:\RR\to\RR$
satisfying $\tilde \g(\theta + 2\pi) = \tilde\g(\theta) + 2\pi$.

Recall that $\vect(S^1)$ is a Lie algebra with the bracket $[f,g]\coloneqq f'g-g'f$.
As in \cite{Tanimoto10-2}, for $0\le k \le \infty$, consider the following Lie subalgebras of $\vect(S^1)$
\begin{align*}
\mathfrak{b}_n&=\left\lbrace f\in C^{\infty}(S^1,\RR): f^{(k)}(-1)=0,\text{for }0\le k\le n\right\rbrace, \\
\mathfrak{b}_{\infty}&=\left\lbrace f\in C^{\infty}(S^1,\RR): f^{(k)}(-1)=0,\text{for all } k\in\NN\right\rbrace.
\end{align*}
To each algebra corresponds a Lie subgroup of $\diff(S^1)$,
\begin{align*}
B_0&\coloneqq\left\lbrace\gamma\in\diff(S^1): \gamma(-1)=-1 \right\rbrace, \\
B_1&\coloneqq\left\lbrace\gamma\in\diff(S^1): \gamma(-1)=-1,\quad \g^{(1)}(-1)=1\right\rbrace, \\
B_n&\coloneqq\left\lbrace\gamma\in\diff(S^1): \gamma(-1)=-1,\quad \g^{(1)}(-1)=1, \g^{(k)}(-1)=0,\text{ for } 2\leq k\leq n \right\rbrace, \\
B_{\infty}&\coloneqq\left\lbrace \gamma\in\diff(S^1): \gamma(-1)=-1, \g^{(1)}(-1)=1, \g^{(k)}(-1)=0,\text{ for all } k\ge 2 \right\rbrace.\end{align*}
By explicit calculations, $\mathfrak b_n$'s are normal Lie subalgebras of $\mathfrak b_0$.
Correspondingly, $B_n$ is a normal subgroup of $B_0$ for every $n\ge 1$: indeed, if $\g\in B_n$ and
$\g_0 \in B_0$, then $\g^{-1}\g_0\g$ has the same $k$-th derivatives at $-1$ as $\g_0$ for $k=1,\cdots, n$,
hence $\g^{-1}\g_0\g\g_0^{-1} \in B_n$ and $\g_0\g\g_0^{-1} \in B_n$.
From this, it is immediate that $\mathfrak b_n$ is a normal Lie subalgebra of $\mathfrak b_1$
and $B_n$ is a normal subgroup of $B_1$.

The quotient Lie algebra $\mathfrak{b}_1/\mathfrak{b}_n$ is finite dimensional and
every element $[g]$ can be identified with the $(n-1)$-tuple of the real numbers $(g^{(2)}(-1),\cdots, g^{(n)}(-1))$.
Furthermore, it follows from straightforward computations that 
$\mathfrak{b}_1/\mathfrak{b}_n$ is nilpotent.
Similarly, the quotient $B_1/B_n$ is a finite-dimensional Lie group and
an element $[\g]\in B_1/B_n$ can be identified with the $(n-1)$-tuple of real numbers $(\g^{(2)}(-1),\cdots, \g^{(n)}(-1))$.
As we see below, $\mathfrak{b}_1/\mathfrak{b}_n$ is the Lie algebra of $B_1/B_n$,
and the latter is connected and simply connected as it is diffeomorphic to $\RR^{n-1}$,
hence the exponential map is surjective \cite[Theorem 1.2.1]{CG90}, which is the key of the following Lemma \ref{lm:inverselimit}.

\begin{lemma}\label{lm:exp}
The Lie algebra of the group $B_1/B_n$ is $\mathfrak{b}_1/\mathfrak{b}_n$.
Let $\Exp_n$ be the exponential map from $\mathfrak{b}_1/\mathfrak{b}_n$ to $B_1/B_n$.
With the natural quotient maps $[\cdot]$, the following diagram commutes.
\[
\begin{tikzcd}
\mathfrak{b}_1 \arrow{r}{[\cdot]} \arrow[swap]{d}{\Exp} & \mathfrak{b}_1/\mathfrak{b}_n \arrow[swap]{d}{\Exp_n}\\
B_1 \arrow{r}{[\cdot]} & B_1/B_n
\end{tikzcd}
\]
\end{lemma}
\begin{proof}
We first prove that if $[f_1] = [f_2]$ in $\mathfrak{b}_1/\mathfrak{b}_n$, then $[\Exp(f_1)] = [\Exp(f_2)]$.
It is enough to show that 
$\Exp(f_1)^{(k)}(-1)=\Exp(f_2)^{(k)}(-1)$ for $0\le k\le n-1$. The case $k=0$ is obvious by definition of $\Exp$.
We show this by induction in $k$.
By \cite[Theorem 7.2]{CL55} we have that
\begin{align}\label{eq:ode}
\frac{\partial }{\partial t}\left(\frac{\partial ^k}{\partial \theta^k}\Exp(tf_j)(e^{i\theta})\bigg\vert_{e^{i\theta}=-1}\right)
=\frac{\partial ^k}{\partial \theta^k}\left(\frac{\partial }{\partial t}\Exp(tf_j)(e^{i\theta})\right)\bigg\vert_{e^{i\theta}=-1}
=\frac{\partial ^k}{\partial \theta^k}f_j(\Exp(tf_j)(e^{i\theta}))\bigg\vert_{e^{i\theta}=-1}
\end{align}
for $j=1,2$.
By the chain rule, we observe that the last expression
can be written in terms of $\frac{\partial ^\ell}{\partial \theta^\ell}\Exp(tf_j)(e^{i\theta})\bigg\vert_{e^{i\theta}=-1}$
with $0 \le \ell \le k-1$ and
$f_j^{(\ell)}(-1)$ with $0\le \ell \le k$, because $f_j^{(1)}(-1) = 0$ and
$\frac{\partial ^k}{\partial \theta^k}\Exp(tf)(e^{i\theta})\bigg\vert_{e^{i\theta}=-1}$ does not appear.
Therefore, 
$\frac{\partial ^k}{\partial \theta^k}\Exp(tf_j)(e^{i\theta})\bigg\vert_{e^{i\theta}=-1}, j=1,2$
satisfy the same differential equation with respect to $t$ \eqref{eq:ode}
with the same initial data, we can conclude that
$\Exp(tf_1)^{(k)}(-1)=\Exp(tf_2)^{(k)}(-1)$.

Let $f \in \mathfrak{b}_1$. Then $[\Exp(tf)]$ is a one-parameter group in $B_1/B_n$
and it does not depend on the representative in $[f]$ by the previous paragraph.
The Lie bracket $[[f],[g]]$ can be computed from $[\Exp(tf)], [\Exp(sg)]$
and it gives $[[f,g]] = [f'g-fg']$, namely, $\mathfrak b_1/\mathfrak b_n$ is the Lie algebra of $B_1/B_n$.
\end{proof}

\begin{lemma}\label{lm:inverselimit}
Let $\{\lambda_n\}_{n\geq 2}$ be a sequence of real numbers.
There exists $g\in C^{\infty}(S^1,\RR)$ such that $\Exp(g)^{(n)}(-1)=\lambda_n $ for all $n\geq 2$.
\end{lemma}
\begin{proof}
Note that $\mathfrak{b}_1/\mathfrak{b}_n\leftarrow \mathfrak{b}_1/\mathfrak{b}_{n+1}$ is a Lie algebra homomorphism,
since $\mathfrak{b}_n\supset\mathfrak{b}_{n+1}$.
Recall the inverse limit of the sequence of Lie algebras
\begin{align*}
\mathfrak{b}_1/\mathfrak{b}_2\longleftarrow\cdots\longleftarrow \mathfrak{b}_1/\mathfrak{b}_n\longleftarrow \mathfrak{b}_1/\mathfrak{b}_{n+1}\longleftarrow\cdots
\end{align*}
is by definition the Lie algebra of sequences $(g_2,\cdots,g_n\dots)$, $g_n \in \mathfrak b_1/\mathfrak b_n$
such that $g_n/\mathfrak b_{n-1} = g_{n-1}$.
This is isomorphic to $\mathfrak b_1/\mathfrak b_\infty$, because any such sequence corresponds to
a sequence $(\lambda'_2, \cdots, \lambda'_n,\cdots)$ where $g_n^{(k)}(-1) = \lambda'_k$
and by Borel's theorem \cite[Theorem 1.2.6]{Hoermander90} there is $g\in\vect(S^1)$
such that $g(-1) = g'(-1) = 0$ and $g^{(k)}(-1) = \lambda'_n$.

Similarly, the inverse limit of the sequence of groups 
\begin{align*}
B_1/B_2\longleftarrow\cdots\longleftarrow B_1/B_n\longleftarrow B_1/B_{n+1}\longleftarrow\cdots
\end{align*}
is isomorphic to $B_1/B_\infty$, because to any sequence $(\lambda_2, \cdots, \lambda_n,\cdots)$
one can associate $\g\in\diff(S^1)$ such that $\g^{(k)}(-1) = \lambda_n$ by Borel's theorem, as we did in Section \ref{nonsmoothsoliton}.

Since $\mathfrak{b}_1/\mathfrak{b}_n$ is a nilpotent Lie algebra,
the exponential map ${\Exp_n:\mathfrak{b}_1/\mathfrak{b}_n\longrightarrow B_1/B_n}$ is surjective \cite[Theorem 1.2.1]{CG90}.
It follows that the inverse limit $\Exp_\infty$ of the maps $\Exp_n$ is surjective,
because any sequence $(\g_2, \cdots, \g_n, \cdots)$ has an inverse image
$(\Exp_2^{-1}(\g_2),\cdots, \Exp_n^{-1}(\g_n),\cdots)$.

To a given sequence $\{\lambda_n\}_{n\geq 2}$, we take the element
$(\lambda_2,\cdots, \lambda_n,\cdots) \in B_1/B_\infty$. Its inverse image with respect
to $\Exp_\infty$ is a sequence $(\lambda'_2,\cdots, \lambda'_n,\cdots) \in \mathfrak b_1/\mathfrak b_\infty$.
By Borel's theorem, there is $g\in\vect(S^1)$ such that
$g(-1) = g'(-1) =0, g^{(n)}(-1) = \lambda_n'$. This $g$ has the desired property
$\Exp(g)^{(k)}(-1) = \lambda_k$ by Lemma \ref{lm:exp}.

\end{proof}
The inverse limit of $\mathfrak{b}_1/\mathfrak{b}_2\leftarrow\cdots\leftarrow \mathfrak{b}_1/\mathfrak{b}_n\leftarrow\cdots$
is isomorphic to the Lie algebra of formal power series $x^2\CC[[x]]$,
where
the Lie bracket is $[f,g]\coloneqq f'g-g'f$, $f,g\in x^2\CC[[x]]$.
Similarly, the inverse limit of $B_1/B_2\leftarrow\cdots\leftarrow B_1/B_n\leftarrow\cdots$
is the group $x+x^2\CC[[x]]$ with product given by the composition of formal power series.

\begin{lemma}\label{lm:glue}
Let $\gamma\in\psone(S^1)$, smooth on $S^1\setminus\{-1\}$ and $\g'(-1) = 1$.
There exist $g$ which is piecewise smooth, $C^1$ and possibly nonsmooth at $\{-1\}$ and $\underline\gamma\in\diff(S^1)$
such that $\gamma=\Exp(g)\circ\underline\gamma$.
\end{lemma}
\begin{proof}
Let us first show that, for two sequences of real numbers $\{\lambda_n^+\}_{n\geq 2}, \{\lambda^-_m\}_{m\geq 2}$,
there exists $g\in C^1(S^1,\mathbb{R})$, such that 
\begin{itemize}
 \item $\Exp(g)(-1) = -1, \Exp(g)^{(1)}(-1) = 1$ 
 \item $\Exp(g)\in\psone(S^1)$
 \item $\Exp(g)$ is smooth on $S^1\setminus\{-1\}$
 \item $\partial_+^n\Exp(g)(-1)=\lambda_n^+ $, $\partial_-^m \Exp(g)(-1)=\lambda_m^- $ for all $n,m\geq2$.
\end{itemize}
By applying Lemma \ref{lm:inverselimit} to $\{\lambda_n^+\}_{n\geq 2}, \{\lambda^-_m\}_{m\geq 2}$,
there exist $g_+,g_-\in C^{\infty}(S^1,\RR)$ such that $\Exp(g_+)^{(n)}(-1)=\{\lambda_n^+\}$
and $\Exp(g_-)^{(n)}(-1)=\{\lambda_n^-\}$, $m,n\ge 2$, $\Exp(g_\pm)^{(n)}(-1) = -1$ and $\Exp(g_\pm)^{(1)}(-1) = 1$.

We may assume that $g_\pm$ has compact support around $-1$. By gluing the restrictions of $g_+$ and $g_-$
to $I_+, I_-$ respectively, we obtain $g$ which is smooth on $S^1\setminus\lbrace -1\rbrace$,
is in $C^1(S^1,\RR)$ and $g|_{I_+}=g_+|_{I_+}$, $g|_{I_-}=g_-|_{I_-}$.
As the only nonsmooth point of $g$ is $-1$, we have $\Exp(g)\in\psone(S^1)$ and
$\Exp(g)$ is smooth on $S^1\setminus\{-1\}$.

For $n\ge 2$, we set $\lambda_n^+ \coloneqq \partial_+^n\gamma(-1)$ and
 $\lambda_m^- \coloneqq \partial_-^m\gamma(-1)$.
 By the observation above, there exists $g\in C^1(S^1,\RR)$, smooth on $S^1\setminus\lbrace-1\rbrace$
 such that $\Exp(g)\in\psone(S^1)$ and $\partial_+^n\Exp(g)(-1)=\lambda_n^+ $, $\partial_-^m \Exp(g)(-1)=\lambda_m^- $
 for all $n,m\geq2$, $\Exp(g)$ is smooth on $S^1\setminus\lbrace-1\rbrace$
 and $\Exp(g)(-1) = -1, \Exp(g)^{(1)}(-1) = 1$.
 It follows that $\underline\gamma\coloneqq \g\circ\Exp(-g)$
 has $\partial^k_+\underline\gamma(-1)=\partial^k_-\underline\gamma(-1)=0$ for all $k\ge 2$, therefore, $\underline\gamma \in B_\infty\subset\diff(S^1)$.
\end{proof}

\begin{proposition}\label{pr:psone}
In any conformal net $(\A, U, \Omega)$, $U$ can be extended to $\psone(S^1)$ (not necessarily continuously)
in such a way that the net is covariant with respect to $U$.
\end{proposition}
\begin{proof}
 We first show that $\g \in \psone(S^1)$ is implementable if $\g(-1) = -1$ and $\g'(-1) = 1$
 and $\g$ is smooth elsewhere.
 
 By Lemma \ref{lm:glue}, we have $\g = \Exp(g)\circ \underline\gamma$, where
 $g$ is piecewise smooth and $C^1$, and $\underline\gamma \in \diff(S^1)$.
 The smooth element $\underline\gamma$ is already implemented by $U$, therefore, to obtain the desired extension,
 it is enough to prove that $\Exp(g)$ is implementable for $g$ which is piecewise smooth and $C^1$.

 Any such $g$ can be approximated by smooth $g(\theta,\mu) = g*h_\mu(\theta)$ with $0<\mu\le 1$, where
 $h_1$ is a smooth function with support in $[-1,1]$ such that $h_1\geq 0$, $\int h_1=1$
 and $h_\mu(\theta) = \frac1\mu h_1\left(\frac\theta \mu\right)$.
 We set $g(\theta,0) = g(\theta)$. Then it is clear that $g$ is a continuous function of $(\theta, \mu)$
 and uniformly Lipschitz in $\theta$, since $\partial_\theta g(\theta,\mu) = g'*h_\mu(\theta)$
 and $\partial_\theta g(\theta,0) = g'(\theta)$. By \cite[Chapter1, Theorem 7.4]{CL55},
 $\Exp(tg_\mu)(\theta)$ is continuous in $\mu$ at each $\t, \theta$.
 Now, as $g_\mu$ is smooth for $\mu > 0$, we have $\Ad U(\Exp(tg_\mu))(\A(I)) = \A(\Exp(tg_\mu)I)$.
 Note that $U(\Exp(tg_\mu)) = e^{itT(g_\mu)}$ up to a scalar, where $T$ is the stress-energy tensor for $U$,
 and since $g_\mu \to g$ in the $\tremezzi$-topology
 \cite[Lemma 4.6]{CW05}, $e^{itT(g_\mu)} \to e^{itT(g)}$ in the strong operator topology
 \cite[Proposition 4.5]{CW05}.
 Therefore, it holds that $\Ad e^{itT(g)}(x) = \lim_{\mu\to 0}\Ad U(\Exp(tg_\mu))(x)$ for any $x \in \A(I)$, and by the continuity above,
 $\Ad e^{itT(g)}(x) \in \A(\Exp(tg)I)$. If we set $U(\Exp(tg)) = e^{itT(g)}$,
this acts covariantly on the net $\A$.
 
 It remains to show that $U$ gives a well-defined projective representation of $\psone(S^1)$.
 Let us first define $U$ and show the well-definedness.
 \begin{itemize}
 \item First consider $\g$ which has only one nonsmooth point at $-1$ and $\g(-1) = -1, \g'(-1) = 1$.
 We take the decomposition $\g = \Exp(g)\circ\underline{\g}$
 and define $U(\g) = U(\Exp(g))U(\underline{\g})$.
 This is well-defined as a projective representation. Indeed, its adjoint action on $\A(I), -1\notin\overline I$
 is determined by $U$ as the representation of $\diff(S^1)$, and such $\A(I)$'s generate
 $\B(\H)$, therefore, if we take another decomposition
 $U(\g) = U(\Exp(g_1))U(\underline{\g_1})$, the difference must be a scalar.
 \item Second, if $\g$ has only one nonsmooth point,
 then $\g = \g_\L\circ \g_0\circ \g_\R$ with $\g_0$ such that $\g_0(-1)=-1, \g_0'(-1) = 1$
 and smooth elements $\g_\L,\g_\R$. By the well-definedness above and the fact that $U$ is
 already defined on smooth elements, $U(\g) = U(\g_\L)U(\g_0)U(\g_\R)$ is also well-defined.
 \item If $\g$ has finitely many nonsmooth points,
 we decompose it into $\g = \g_0\cdot \check \g$, where
 $\g_0$ fixes all these nonsmooth points, has derivative $1$ and $\supp \g_0$
 is a disjoint union of intervals around these nonsmooth points.
 As the nonsmooth part has disjoint unions, if we take the product
 of $U$ defined above on each component, this does not depend on the order of the product.
 If we consider two such decompositions, the nonsmooth parts cancel each other
 up to a smooth element which is already defined, hence $U$ is well-defined.
 \end{itemize}
 That $U$ is a projective representation is shown as follows.
 \begin{itemize}
  \item For two elements $\g_1,\g_2$ such that $\g_j(-1) = -1$,
  $\Ad U(\g_1\g_2)$ and $\Ad U(\g_1)U(\g_2)$ implement the same action of $\A(I)$ such
  that $-1\notin\overline I$ as before, hence the difference between $U(\g_1\g_2)$ and $U(\g_1)U(\g_2)$
  must be a scalar.
  \item By rotation, the homomorphism property $U(\g_1\g_2) = U(\g_1)U(\g_2)$, up to a scalar, follows also
  when $\g_1$ and $\g_2$ has only one and same nonsmooth point.
  \item For two elements $\g_1,\g_2$ with finitely many nonsmooth points,
  take the decompositions as above: $\g_j = \g_{j,0},\check \g_j$.
  We may assume that the components of $\check \g_1\g_{2,0}\check\g_1^{-1}$
  is either disjoint from the components of $\g_{1,0}$, or have a common nonsmooth point.
  If they are disjoint, their representation by $U$ commute. If they have a common nonsmooth
  point, we can merge them to a single element and we have shown the homomorphism property 
  above. In this way, we have the decomposition
  $\g_1\g_2 = (\g_{1,0}\check \g_1\g_{2,0}\check\g_1^{-1})\cdot (\check \g_1\check\g_2)$,
  where $\g_{1,0}\check \g_1\g_{2,0}\check\g_1^{-1}$ is supported around the nonsmooth points
  and $\check \g_1\check\g_2$ is smooth and we have
  $U(\g_1)U(\g_2) = U(\g_{1,0})U(\check \g_1)U(\g_{2,0})U(\check\g_2) = U(\g_{1,0}\check \g_1\g_{2,0}\check\g_1^{-1})U(\check \g_1\check\g_2) = U(\g_1\g_2)$ up to a scalar.  
 \end{itemize}
 We have seen the covariance of the net with respect to $\g_0$ such that $\g_0(-1) = -1$
 and $\g_0'(-1) = 1$ and smooth elsewhere.
 Any element $\g \in \psone(S^1)$ can be decomposed as a product of such elements
 and smooth elements, and for each of them we have shown the covariance,
 hence the covariance holds also for $\g$. 
\end{proof}

\subsubsection*{Acknowledgements.}
We would like to thank Roberto Longo for suggesting the problem.
We are grateful to Sebastiano Carpi, Andr\'e Henriques, Karl-Hermann Neeb and Stefano Rossi for various interesting discussions.
S.D.\! and Y.T.\! acknowledge the MIUR Excellence Department Project awarded to the Department of Mathematics, University of Rome Tor Vergata, CUP E83C18000100006.

{\small
\newcommand{\etalchar}[1]{$^{#1}$}
\def\cprime{$'$} \def\polhk#1{\setbox0=\hbox{#1}{\ooalign{\hidewidth
  \lower1.5ex\hbox{`}\hidewidth\crcr\unhbox0}}} \def\cprime{$'$}

}
\end{document}